\documentclass[12pt,letterpaper]{article}

\usepackage[letterpaper, left=1in, right=1in, top=0.9in, bottom=0.9in]{geometry}
\usepackage[american]{babel}
\usepackage[normalem]{ulem}
\usepackage{amsmath, amssymb, cases, amsthm}
\usepackage{thmtools}
\usepackage[shortlabels]{enumitem}
\usepackage{mdframed}
\usepackage{bbm}
\usepackage{bm}
\usepackage{microtype}
\usepackage{xcolor}
\usepackage{makecell}
\usepackage{mathtools}
\usepackage{algorithmic}
\usepackage[procnumbered,ruled,vlined,linesnumbered]{algorithm2e}
\usepackage{float}
\usepackage{varwidth}
\usepackage{modletter}
\usepackage{tcolorbox}
\usepackage{tikz}
\usetikzlibrary{arrows.meta,positioning,fit,calc}

\newtcolorbox{construction}[2][]
{
	colframe = gray!50,
	colback  = gray!10,
	coltitle = gray!10!black,
	left*=0mm, 
	before skip = 10pt,
	after skip = 10pt,
	title    = \textbf{\space\space #2},
	#1,
}

\SetKwInput{KwData}{Input}
\SetKwInput{KwResult}{Output}
\SetKwInput{KwGlobalVar}{Global variables}

\usepackage{xcolor}
\definecolor{ForestGreen}{rgb}{0.1333,0.5451,0.1333}
\definecolor{DarkRed}{rgb}{0.80,0,0}
\definecolor{Red}{rgb}{1,0,0}
\usepackage[linktocpage=true,
pagebackref=true,colorlinks,
linkcolor=DarkRed,citecolor=ForestGreen,
bookmarks,bookmarksopen,bookmarksnumbered]
{hyperref}

\usepackage[capitalize,nosort,nameinlink]{cleveref}

\declaretheorem[numberwithin=section,refname={Theorem,Theorems},Refname={Theorem,Theorems}]{theorem}
\declaretheorem[numberwithin=section,name=Theorem,refname={Theorem,Theorems},Refname={Theorem,Theorems}]{thm}
\declaretheorem[numberlike=theorem]{lemma}

\declaretheorem[numberlike=theorem,style=definition]{definition}

\declaretheorem[numberlike=theorem,style=remark]{remark}

\declaretheorem[numberlike=theorem, refname={Problem,Problems},Refname={Problem,Problems},name={Problem}]{problem}

\theoremstyle{definition}

\def\final{1}  %
\ifnum\final=0  %
\newcommand{\todo}[1]{{\color{red}[{\tiny TODO: \bf #1}]\marginpar{\color{red}*}}}
\newcommand{\thatchaphol}[1]{{\color{purple}[{\tiny Thatchaphol: \bf #1}]\marginpar{\color{purple}*}}}
\newcommand{\yonggang}[1]{{\color{blue}[{\tiny Yonggang: \bf #1}]\marginpar{\color{blue}*}}}
\else %
\newcommand{\yonggang}[1]{}
\newcommand{\todo}[1]{}
\newcommand{\thatchaphol}[1]{}
\fi

\newcommand{\eps}{\epsilon}
\newcommand{\poly}{\mathrm{poly}}
\newcommand{\polylog}{\mathrm{polylog}}
\newcommand{\dist}{\mathrm{dist}}

\newcommand{\Bo}{\mathrm{Ball}^{\mathrm{out}}}
\newcommand{\Bi}{\mathrm{Ball}^{\mathrm{in}}}
\newcommand{\dego}{\deg^{\mathrm{out}}}
\newcommand{\degi}{\deg^{\mathrm{in}}}
\newcommand{\Start}{\mathrm{Start}}
\newcommand{\End}{\mathrm{End}}

\newcommand{\Ho}{H^{\mathrm{out}}}
\newcommand{\Hi}{H^{\mathrm{in}}}
\newcommand{\Hm}{H^{\mathrm{mid}}}
\newcommand{\bHo}{\bar{H}^{\mathrm{out}}}
\newcommand{\tHo}{\tilde{H}^{\mathrm{out}}}
\newcommand{\bHi}{\bar{H}^{\mathrm{in}}}
\newcommand{\iBo}{B^{\mathrm{out}}}
\newcommand{\bBo}{\bar{B}^{\mathrm{out}}}
\newcommand{\iBi}{B^{\mathrm{in}}}
\newcommand{\bBi}{\bar{B}^{\mathrm{in}}}
\newcommand{\GC}{G'_{\mathrm{SCC}}}

\usepackage{tikz}
\usetikzlibrary{arrows.meta,positioning,calc,decorations.pathreplacing}
\usepackage{subcaption} %

\usetikzlibrary{positioning,calc,arrows.meta,bending}
\begin{document}
\sloppy

\title{
Deterministic Negative-Weight Shortest Paths\\ in Nearly Linear Time via Path Covers
}
\author{Bernhard Haeupler\thanks{
        INSAIT, Sofia University ``St.~Kliment Ohridski'' and ETH Zürich,
        \texttt{bernhard.haeupler@insait.ai}.
        Partially funded by the Ministry of Education and Science of Bulgaria's support for INSAIT as part of the Bulgarian National Roadmap for Research Infrastructure and through the European Research Council (ERC) under the European Union's Horizon 2020 research and innovation program (ERC grant agreement 949272).} \and 
Yonggang Jiang\thanks{MPI-INF and Saarland University, Germany, \texttt{yjiang@mpi-inf.mpg.de}. Part of this work was done while visiting INSAIT. Supported by Google PhD Fellowship.} \and   
Thatchaphol Saranurak\thanks{
        University of Michigan,
        \texttt{thsa@umich.edu}.
        Supported by NSF Grant CCF-2238138 and a Sloan Fellowship. Part of this work was done at INSAIT. Partially funded by the Ministry of Education and Science of Bulgaria's support for INSAIT as part of the Bulgarian National Roadmap for Research Infrastructure. }
}
\date{}
\maketitle
\pagenumbering{gobble}
\begin{abstract}

We present the first \emph{deterministic} nearly-linear time algorithm for single-source shortest paths with negative edge weights on directed graphs: given a directed graph $G$ with $n$ vertices, $m$ edges whose weights are integers in $\{-W,\dots,W\}$, our algorithm either computes all distances from a source $s$ or reports a negative cycle in $\tilde{O}(m)\cdot \log(nW)$ time. 

All  known near-linear time algorithms for this problem have been inherently randomized \cite{BernsteinNW25,BringmannCF23,FischerHLRS25,li2024bottom}, as they crucially rely on \emph{low-diameter decompositions}.

To overcome this barrier, we introduce a new structural primitive for directed graphs called the \emph{path cover}. This plays a role analogous to \emph{neighborhood covers} in undirected graphs \cite{awerbuch1998near}, which have long been central to derandomizing algorithms that use low-diameter decomposition in the undirected setting. We believe that path covers will serve as a fundamental tool for the design of future deterministic algorithms on directed graphs.

\end{abstract}

\clearpage
\tableofcontents

\clearpage

\pagenumbering{arabic}
\section{Introduction}

In the negative weight single-source shortest paths (SSSP) problem, given an $m$-edge $n$-vertex directed weighted graph $G=(V,E,w)$ with \emph{possibly negative} edge weight $w(e)\in\{-W,\dots,W-1,W\}$ for each $e\in E$ and a source vertex $s\in V$, the goal is to either compute the distance from $s$ to every other vertex or find a negative-weight cycle.

When weights are non-negative, the textbook Dijkstra algorithm takes near-linear $O(m+n\log n)$ time. Another textbook Bellman-Ford algorithm can handle negative weights, but requires much slower $O(mn)$ time. This contrast motivates the search for near-linear time algorithms since 1980s \cite{garbow1985scaling,gabow1989faster,goldberg1995scaling,cohen2017negative,van2020bipartite,axiotis2020circulation,van2021minimum}.

In 2022, two algorithmic breakthroughs happened. Chen et al.~\cite{ChenKLPGS22} gave a randomized algorithm with $m^{1+o(1)}\log W$ time that solves an even more general minimum cost flow problem. Their approach is based on convex optimization and dynamic data structures and was later improved to be deterministic \cite{van2023deterministic,chen2024almost,van2024almost}. However, the dynamic data structures required in this approach are complicated and the $m^{o(1)}$ factor in time seems hard to avoid.

Independently, in the same year, Bernstein, Nanongkai, and Wulff-Nilsen \cite{BernsteinNW25} showed an algorithm with $O(m\log^{8}(n)\log(W))$ time, which is faster, combinatorial, and bypasses complicated dynamic data structures. The follow-up work \cite{BringmannCF23,ashvinkumar2023parallel,li2024bottom,FischerHLRS25} subsequently sped up the time to $O((m+n\log\log n)\log^{2}(n)\log(nW))$. Unfortunately, algorithms in this line of work are inherently randomized.

Until now, there is no deterministic negative weight SSSP algorithm with near-linear $O(m\cdot\mathrm{polylog}(nW))$ time.

\paragraph{Our result.}

We give the first deterministic near-linear algorithm, improving upon both lines of work above. 
\begin{thm}
[Deterministic SSSP] \label{thm:main} There is a \emph{deterministic} algorithm solving the negative weight single-source shortest path problem in $O(m\log^{8}n\log(nW))$ time. 
\end{thm}

To prove \Cref{thm:main}, we introduce a novel graph structure called \emph{path cover}, which we believe will find further applications beyond \Cref{thm:main} itself. We do not optimize the polylogarithmic factor in order to present the clearest version of this new concept. A more relaxed version of path covers can speed up our algorithm to $O(m\log^{5}n\log(nW))$ time. Below, we first define the path covering problem to motivate path covers.

\paragraph{The Path Covering Problem.}

We say a path $p$ is a \emph{$d$-path} if its length is at most $d$, and a subgraph $H\subseteq G$ \emph{covers} $p$ if $p\subseteq E(H)$. A graph $H$ is \emph{$d$-clustered} if every strongly-connected component of $H$ has diameter at most $d$. The negative-weight SSSP problem can be reduced to the following problems. 
\begin{problem}
[Covering $d$-Paths with Clustered Subgraphs]\label{prob:path covering}Fix a graph $G$ with \emph{non-negative} weights and $d\ge0$. An adversary chooses an \emph{unknown target} set $P$ of paths of length at most $d$ in $G$ where $|P|=O(n)$. We must choose $\tilde{O}(d)$-clustered subgraphs $G'_{1},\dots,G'_{z}\subseteq G$ where $z=O(\log n)$ such that each path $p\in P$ is covered by some $G'_{i}$. 
\end{problem}

The frameworks of \cite{ashvinkumar2023parallel} and \cite{FischerHLRS25} showed that a deterministic near-linear time algorithm for \Cref{prob:path covering} implies another one for negative-weight SSSP with $m^{1+o(1)}\log(nW)$ time and $O(m\cdot\mathrm{polylog}(nW))$ time, respectively.

\paragraph{Previous Solutions: Either Randomized or Undirected.}

\Cref{prob:path covering} admits a randomized algorithm with near-linear time via the probabilistic \emph{low-diameter decomposition} (LDD), a powerful primitive \cite{bartal1996probabilistic,miller2013parallel} extended to directed graphs by \cite{BernsteinNW25}. For any $d$, an LDD algorithm samples a $\tilde{O}(d)$-clustered subgraph $G'\subseteq G$ that covers each $d$-path with constnat probability. Thus, $z=O(\log n)$ samples $G'_{1},\dots,G'_{z}$ will cover $|P|=\poly(n)$ unknown target paths with high probability, solving \Cref{prob:path covering}. All prior negative-weight SSSP algorithms with near-linear time \cite{BernsteinNW25,BringmannCF23,ashvinkumar2023parallel,li2024bottom,FischerHLRS25} employ probabilistic LDDs, and, hence, are inherently randomized because of it. 

For deterministic algorithms, \Cref{prob:path covering} is much more difficult. Observe that we must cover \emph{all} exponentially many $d$-paths, as any uncovered path can be picked by the adversary into the target set $P$. Thus, even the \emph{existence} of the solution is highly unclear. 

Yet, in undirected graphs, a deterministic near-linear time algorithm follows directly from efficient constructions of \emph{neighborhood covers} \cite{awerbuch1998near}. This very strong property of neighborhood covers has been used for developing many deterministic algorithms in undirected graphs \cite{awerbuch1998near,ghaffari2019improved,chuzhoy2021decremental,chuzhoy2023new,haeupler2024dynamic,chuzhoy2025fully}. For completeness, we define LDD and neighborhood covers, and explain how they give the solution of \Cref{prob:path covering} in either the randomized or undirected settings in \Cref{sec:previous solution}.

This raises a natural question: is there a \emph{deterministic} algorithm on \emph{directed} graphs for \Cref{prob:path covering}? 

\paragraph{A Barrier for Subgraphs: No Deterministic Algorithm on Directed Graphs.}

Unfortunately, unlike undirected graphs, we show that \Cref{prob:path covering} \emph{cannot} be solved deterministically on directed graphs.
In fact, even the \emph{total size} of the $\tilde{O}(d)$-clustered subgraphs must be $\tilde{\Omega}(m\sqrt{m})$ in the worst case. Thus, it takes $\tilde{\Omega}(m\sqrt{m})$ time to even write down these subgraphs.

\begin{restatable}[Barrier for Subgraphs] {thm}{barriar}
\label{thm:barrier}For every $m,\lambda\ge1$, there exists an $O(m)$-edge directed graph $G$ and $d$ such that every collection of $d\lambda$-clustered graphs $G'_{1},\dots,G'_{z}\subseteq G$ covering all $d$-paths in $G$ must have total size $\sum_{i}|E(G'_{i})|=\Omega(m\sqrt{m/\lambda})$. 
\end{restatable}

See \Cref{sec:barrier} for the proof.
This barrier seems bad because 
the only two main tools mentioned above, LDDs and neighborhood covers, naturally only give a collection of $\tilde{O}(d)$-clustered \emph{subgraphs} of $G$ and, hence, must suffer from the barrier of \Cref{thm:barrier}.

\paragraph{Bypassing the Barrier via Projection.}

We bypass the barrier in \Cref{thm:barrier} using a \emph{projection}, instead of subgraphs. A \emph{projection} to $G$ is a graph $G'$ whose vertices of $G'$ are copies of vertices in $G$, and edges of $G'$ may connect the copy $u'$ to $v'$ only if the original vertex $u$ is connected to $v$ in $G$. More formally, there is a mapping $\pi:V(G')\rightarrow V(G)$ where, if $(u',v')\in E(G')$, then $(u=\pi(u'),v=\pi(v'))\in E(G)$.\footnote{That is, $\pi$ is a graph homomorphism.} Intuitively, a projection is formed by ``merging'' subgraphs together. 

Observe that every path $p'=(v'_{1},\dots,v'_{t})$ in $G'$ is mapped to a path $p=(v_{1},\dots v_{t})$ in $G$ where $v_{i}=\pi(v'_{i})$. Here, we say that $p$ is \emph{covered }by $G'$ and $p'$ is a \emph{lift} of $p$. If a projection $G'$ covers all $d$-paths in $G$, then $G'$ is a \emph{$d$-path cover} of $G$. 

As a toy example to illustrate the power of projection, consider a directed unweighted cycle $G$ on $[n]$. 
Observe that a path $G'=(1,\dots,n,1,\dots n)$ is an $n$-path cover of $G$, which is a $0$-clustered graph, i.e., a DAG. In contrast, any collection of $0$-clustered subgraphs $G'_{1},\dots,G'_{z}\subseteq G$ cannot cover any $n$-path of $G$. Otherwise, it would introduce a cycle in $G'_{i}$. See \Cref{fig:projection} for another example where subgraphs require larger total size than a projection.

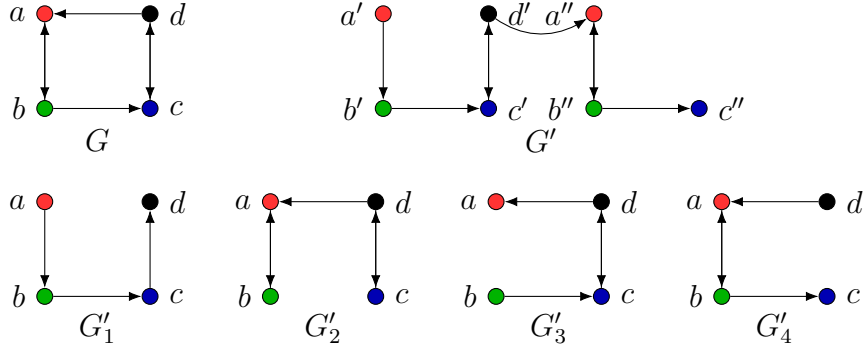
\begin{figure}

\tikzset{
	dot/.style={circle,draw,minimum size=6pt,inner sep=0pt},
	a/.style={dot,fill=red!80},
	b/.style={dot,fill=green!70!black},
	c/.style={dot,fill=blue!70!black},
	d/.style={dot,fill=black},
	>={Latex}
}

\newcommand{\AB}[3]{%
	\node[a,label=left:{$a$}] (#3a) at (#1,#2) {};
	\node[b,label=left:{$b$}] (#3b) at (#1,#2-0.9) {};
	\draw[->] (#3a) -- (#3b);
	\draw[->] (#3b) -- (#3a);
}
\newcommand{\ABname}[5]{%
	\node[a,label=left:{#4}] (#3a) at (#1,#2) {};
	\node[b,label=left:{#5}] (#3b) at (#1,#2-0.9) {};
	\draw[->] (#3a) -- (#3b);
	\draw[->] (#3b) -- (#3a);
}
\newcommand{\DC}[3]{%
	\node[d,label=right:{$d$}] (#3d) at (#1,#2) {};
	\node[c,label=right:{$c$}] (#3c) at (#1,#2-0.9) {};
	\draw[->] (#3d) -- (#3c);
	\draw[->] (#3c) -- (#3d);
}
\newcommand{\DCname}[5]{%
	\node[d,label=right:{#4}] (#3d) at (#1,#2) {};
	\node[c,label=right:{#5}] (#3c) at (#1,#2-0.9) {};
	\draw[->] (#3d) -- (#3c);
	\draw[->] (#3c) -- (#3d);
}

\newcommand{\ABone}[3]{%
	\node[a,label=left:{$a$}] (#3a) at (#1,#2) {};
	\node[b,label=left:{$b$}] (#3b) at (#1,#2-0.9) {};
	\draw[->] (#3a) -- (#3b);
}
\newcommand{\ABoneName}[5]{%
	\node[a,label=left:{#4}] (#3a) at (#1,#2) {};
	\node[b,label=left:{#5}] (#3b) at (#1,#2-0.9) {};
	\draw[->] (#3a) -- (#3b);
}
\newcommand{\DCcd}[3]{%
	\node[d,label=right:{$d$}] (#3d) at (#1,#2) {};
	\node[c,label=right:{$c$}] (#3c) at (#1,#2-0.9) {};
	\draw[->] (#3c) -- (#3d);
}

\newcommand{\Aonly}[3]{\node[a,label=left:{$a$}] (#3) at (#1,#2) {};}
\newcommand{\Bonly}[3]{\node[b,label=left:{$b$}] (#3) at (#1,#2) {};}
\newcommand{\Conly}[3]{\node[c,label=right:{$c$}] (#3) at (#1,#2) {};}
\newcommand{\Donly}[3]{\node[d,label=right:{$d$}] (#3) at (#1,#2) {};}

\centering{
	\begin{tikzpicture}[x=1.4cm,y=1.4cm]
		
		\begin{scope}[xshift=0cm,yshift=0cm]
			\AB{0}{1}{L}
			\DC{1}{1}{R}
			\draw[->] (Lb) -- (Rc);                 %
			\draw[->] (Rd) -- (La);                 %
			\node at (0.5,-0.2) {$G$};
		\end{scope}
		
		\begin{scope}[xshift=4.5cm,yshift=0cm]
			\ABoneName{0}{1}{Aone}{$a'$}{$b'$}      %
			\DCname{1}{1}{Mid}{$d'$}{$c'$}          %
			\ABname{2}{1}{Atwo}{$a''$}{$b''$}       %
			\node[c,label=right:{$c''$}] (Ctwo) at (3,0.1) {};
			\draw[->] (Aoneb) -- (Midc);            %
			\draw[->,bend left=-35] (Midd) to (Atwoa); %
			\draw[->] (Atwob) -- (Ctwo);            %
			\node at (1.5,-0.2) {$G'$};
		\end{scope}
		
		\begin{scope}[yshift=-2.5cm]
			\begin{scope}[xshift=0cm]
				\ABone{0}{1}{L}
				\DCcd{1}{1}{R}
				\draw[->] (Lb) -- (Rc);               %
				\node at (0.5,-0.2) {$G'_1$};
			\end{scope}
			
			\begin{scope}[xshift=3cm]
				\AB{0}{1}{L}
				\DC{1}{1}{R}
				\draw[->] (Rd) -- (La);               %
				\node at (0.5,-0.2) {$G'_2$};
			\end{scope}
			
			\begin{scope}[xshift=6cm]
				\Aonly{0}{1}{aonly}
				\Bonly{0}{0.1}{bonly}
				\DC{1}{1}{R}
				\draw[->] (bonly) -- (Rc);            %
				\draw[->] (Rd) -- (aonly);            %
				\node at (0.5,-0.2) {$G'_3$};
			\end{scope}
			
			\begin{scope}[xshift=9cm]
				\AB{0}{1}{L}
				\Conly{1}{0.1}{conly}
				\Donly{1}{1}{donly}
				\draw[->] (Lb) -- (conly);            %
				\draw[->] (donly) -- (La);            %
				\node at (0.5,-0.2) {$G'_4$};
			\end{scope}
		\end{scope}
		
	\end{tikzpicture}
}

\caption{At the top right, $G'$ is a $1$-clustered \emph{projection} to $G$ that covers all $3$-paths in $G$.
At the second row,
$G'_1,G'_2,G'_3,G'_4$ are $1$-clustered \emph{subgraphs} of $G$ that cover all $3$-paths in $G$.}
\label{fig:projection}
\end{figure}

As our main technical contribution, we show how projection can bypass the the barrier from \Cref{thm:barrier} for subgraphs: while some graph requires $\tilde{O}(d)$-clustered subgraphs of total size $\tilde{\Omega}(m\sqrt{m})$ to cover all $d$-paths, we show that every graph $G$ admits a $\tilde{O}(d)$-clustered projection $G'$ to $G$ of linear size that covers all $d$-paths.

\begin{thm}
[Covering $d$-Paths with Clustered Projection]\label{thm:path cover intro}Given a graph $G$ with \emph{non-negative} weights and $d\ge0$, there exists a projection $G'$ to $G$ where $G'$ is a $O(d\log^{6}n)$-clustered graph of size $|E(G')|\le(1+\frac{1}{\log n})|E(G)|$ that covers all $d$-paths of $G$. Moreover, $G'$ can be deterministically computed in $O(m\log^{4}n)$ time.
\end{thm}

In other words, $G'$ is a $d$-path cover which is a $\tilde{O}(d)$-clustered graph of size $(1+\frac{1}{\log n})|E(G)|$. Importantly, this size bound is even stronger than that of the known randomized LDD-based approach, which requires the total size of $O(|E(G)|\log n)$.

Finally, we obtain our SSSP algorithm \Cref{thm:main} by showing that the reduction of \cite{ashvinkumar2023parallel} from negative-weight SSSP to \Cref{prob:path covering} extends from subgraphs to projections.
By replacing the randomized LDD with the deterministic projection from  \Cref{thm:path cover intro}, we simultaneously derandomize their algorithm and speed up the time from $m^{1+o(1)}$ to near-linear. 

This speed-up crucially exploits our size bound of $(1+\frac{1}{\log n})|E(G)|$. Thus, when the algorithm makes $O(\log n)$-depth recursion, the total size blow-up is at most a constant factor as $(1+\frac{1}{\log n})^{O(\log n)}|E(G)|=O(|E(G)|)$.
In contrast, the LDD-based approach gives subgraphs of total size $O(|E(G)|\log n)$. The $O(\log n)$ blow-up per level causes their $m^{1+o(1)}$ running time in \cite{ashvinkumar2023parallel}.\footnote{The projection from \Cref{thm:path cover intro} can derandomize \cite{FischerHLRS25}, too. Here, a more relaxed size bound of $|E(G')|=O(|E(G)| \log n)$ is sufficient but 
their framework additionally requires that every SCC of $G'$ contains no duplicated copies of vertices. Although our projection $G'$ \emph{does} satisfy this property, the overall argument is more complicated, and we do not want to impose such restrictions on the definition of path cover. So, we choose to present a derandomization of  \cite{ashvinkumar2023parallel} in this paper.}

Path covers can be viewed as a deterministic counterpart of LDD in directed graphs, similar to how neighborhood cover is a deterministic counterpart of LDD in undirected graphs. Since neighborhood covers have been a powerful tool in undirected graphs, we believe \Cref{thm:main} is only the first of many applications of path covers.

\paragraph{Related works.}
In the setting of  \emph{undirected} graphs, the idea of simplifying a graph $G$ via another graph $G'$, that makes copies of the vertices of the original graph $G$ and connects them together, appeared in the \emph{clan embedding} of \cite{filtser2021clan} and the \emph{copy tree embedding} of
\cite{haepler2022adaptive}.

\paragraph{Independent work.}
In an independent work, Li~\cite{li25_det} also obtains a near-linear time deterministic algorithm for negative-weight shortest paths. He defines a new variant of \emph{padded decomposition} which can be of independent interest. In constructing our path cover, we perform a similar graph decomposition.
\section{Overview}

We start this section by sketching the \emph{existence} of the $d$-path cover $G'$ of $G$ in \Cref{thm:path cover intro}.%

\smallskip
We start with some standard notation: For every vertex $v\in V$ let $\dego_G(v)$ be the out-degree of $v$, i.e., the number of edges leaving $v$. For a vertex set $S\subseteq V$ let $\dego_G(S)=\sum_{v\in S}\dego_G(v)$ be the number of edges starting in $S$. Lastly, for any vertex $u\in V$, define
\[
\Bo_{G}(u,r)=\{v\in V\mid \dist_G(u,v)\le r\}
\]
to be the ball centered at $u$ with radius $r$.

\paragraph{Construction.}

Our algorithm is recursive. 

If the diameter of $G$ is one strongly connected component with diameter $O(d\log^3 n)$, then one can trivially return the whole graph and terminate the recursion. 

Otherwise, the algorithm finds a good partition by picking an arbitrary node $u$ and growing a ball from it by computing $\Bo_G(u,i\cdot d)$ for $i=1,2,3,\ldots$ until we find a ``thin layer'' $i^*$ such that
\begin{equation}\label{eq:overview grow balls}
    \dego_G\!\left(\Bo_G(u,i^*\cdot d)\right)\le \left(1+\frac{1}{\log^2 n}\right)\cdot \dego_G\!\left(\Bo_G(u,(i^*-1)\cdot d)\right)\, .
\end{equation}
Since the ball cannot grow multiplicatively too often it is clear that $i^*=O(\log^3 n)$.
The algorithm works separately and recursively on both
\[
\iBo:=\Bo_G(u,i^*\cdot d)\qquad\text{and}\qquad\bBo:=V\setminus \Bo_G\!\left(u,(i^*-1)\cdot d\right).
\]
Note that $\iBo$ and $\bBo$ only \emph{overlap} at the ``ring'' around $u$ from radius $(i^*-1)\cdot d$ to $i^*\cdot d$.

We focus here on the interesting case that the ball grown from $u$ is not more than a constant fraction of the graph, i.e., $\dego_G\!\left(\iBo\right)\le 0.9m$, our proof shows that otherwise there is a way to carve out a large ball and still make progress. Note that by running Dijkstra's algorithm locally from $u$ the time needed to grow the ball $\iBo$ is $\tO{\dego_G(\iBo)}$ and therefore essentially proportional to its size and independent of the size of $\bBo$. This makes the recursion efficient. Let $\Ho$ and $\bHo$ be $d$-path covers of the induced subgraphs $G[\iBo]$ and $G[\bBo]$, respectively, which are constructed \emph{recursively}. 

For technical reasons, we strengthen the property of $\Ho$ (and likewise for $\bHo$): not only do we require each $d$-path in $G[\iBo]$ to have a lift $p'$ in $\Ho$, we require a \emph{one-to-one} mapping between the start vertices of $p$ and $p'$. In other words, each vertex $v$ in $G[\iBo]$ has a \emph{representative} vertex $v'$ in $\Ho$ such that if a $d$-path in $G$ starts at $v$, then its lift in $G'$ starts at $v'$.

We build the $d$-path cover of $G$ (denoted $G'$) by concatenating $\bHo$ and $\Ho$ as follows: for every $(u,v)\in E$ with $u\in \bBo$ and $v\in V\setminus\bBo$, we add edges from all copies of the vertex $u$ in $\bHo$ to the \emph{representative} of the vertex $v$ in $\Ho$. Having a representative for each vertex avoids having to add prohibitively many edges from all copies of $u$ to all copies of $v$. Defining the representatives of vertices in $G'$ throughout the recursion is straightforward once the path-covering property in the next paragraph is clear.

\begin{figure}[H]
  \centering
  \includegraphics[width=0.5\textwidth]{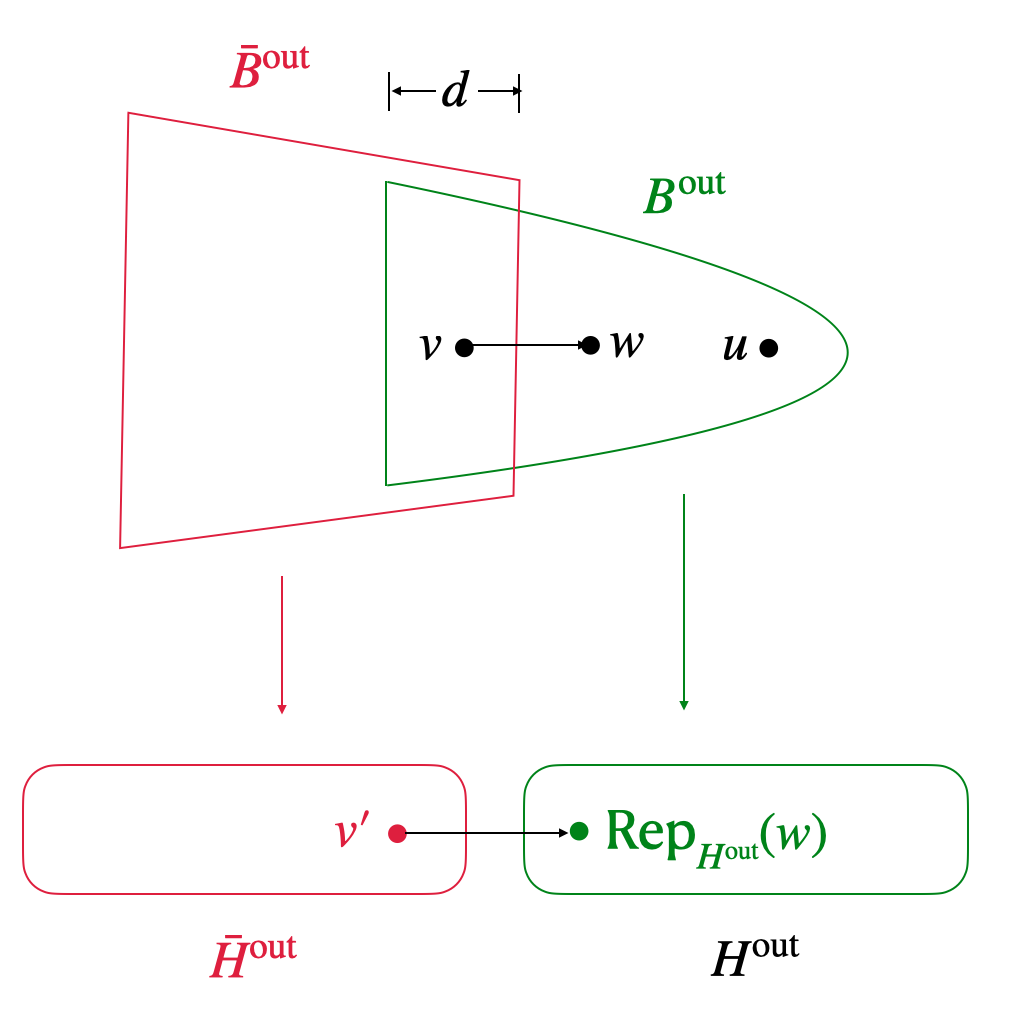} %
  \caption{An illustration of the construction. A path entering $V-\bBo$ (for example a path starting at $w$) will not leave $\iBo$ again because of the length $d$ layer in the middle, thus covered by $\Ho$.}
  \label{fig:example}
\end{figure}

\paragraph{$\tO{d}$-clustered graph.} Notice that if $\bHo,\Ho$ are $\tO{d}$-clustered graphs, then $G'$ is a $\tO{d}$-clustered graph because edges can only be from $\bHo$ to $\Ho$. %

\paragraph{Path Covering.}
We show that $G'$ indeed covers all $d$-paths in $G$.

We consider the case where a path $p$ starts at $v_p\in \bBo$ and leaves $\bBo$ at some point (i.e., enters $V\setminus\bBo$). Other situations (completely inside $\bBo$ or starting in $V\setminus\bBo$) can be handled similarly.
Let $v'_p$ be the first vertex on $p$ that is not in $\bBo$.
Let $p_1$ be the subpath of $p$ before $v'_p$ (excluding $v'_p$) and let $p_2$ be the remainder of $p$ (starting at $v'_p$).

Notice that
\[
v'_p\in V\setminus\bBo
=V\setminus\!\left(V\setminus\Bo_G\!\left(u,(i^*-1)\cdot d\right)\right)
=\Bo_G\!\left(u,(i^*-1)\cdot d\right).
\]
This means that every path starting at $v'_p$ with length at most $d$ (in particular, $p_2$) is contained in
\[
\Bo_G\!\left(u,(i^*-1)\cdot d + d\right)
=\Bo_G\!\left(u,i^*\cdot d\right)
=\iBo.
\]
Therefore $p_1$ lies completely inside $G[\bBo]$ and $p_2$ lies completely inside $G[\iBo]$, so $\bHo$ covers $p_1$ with a lift $p'_1$ and $\Ho$ covers $p_2$ with a lift $p'_2$.

We then concatenate $p'_1$ and $p'_2$ using the edge from the last vertex of $p'_1$ to the first vertex of $p'_2$ (which is the representative of $v'_p$ in $\Ho$). The concatenated path maps to $p$.

\paragraph{Size of the path cover.}
We first show that the size of $G'$ satisfies
\[
|E(G')|\le \sum_{v'\in V(G')} \dego_G\!\bigl(\pi(v')\bigr).
\]
Let $(u',v')$ be an edge concatenating $\bHo$ and $\Ho$, where $u'\in V(\bHo)$ and $v'\in V(\Ho)$ are copies of $u\in \bBo$ and $v\in V\setminus\bBo$, respectively, so that $(u,v)\in E$ and $v'$ is the representative of $v$ in $\Ho$. In this case, we say that \emph{$(u,v)$ contributes to $u'$}.
Notice that (i) $(u,v)$ can contribute to $u'$ at most once among all edges concatenating $\bHo$ and $\Ho$, because we add only one edge from $u'$ to the representative of $v$ in $\Ho$ (which is why we define representatives); and (ii) $(u,v)$ cannot contribute to $u'$ again in the recursive construction of $\bHo$, because $(u,v)$ is not an edge of $G[\bBo]$. Combining these two facts, every vertex $u'$ is contributed to by at most $\dego_G\!\bigl(\pi(u')\bigr)$ edges, which establishes the claimed inequality.

It remains to bound $\sum_{v'\in V(G')} \dego_G\!\bigl(\pi(v')\bigr)$. Intuitively, \Cref{eq:overview grow balls} guarantees that each recursion layer increases the total degree by at most a factor of $\bigl(1+1/\log^2 n\bigr)$ times the smaller side which is sufficient to yields an overall bound of $(1+1/\log n)\cdot m$ on the final total degree (we defer the calculation to the technical section). Since we have $|E(G')|\le \sum_{v'\in V(G')} \dego_G\!\bigl(\pi(v')\bigr)$, we get the desired size bound
\[|E(G')|\le (1+1/\log n)\cdot m\]

\section{Preliminaries}
We use the following notations, $[a]=\{1,2,\dots,a\}$, $\tO{f}=f\cdot \polylog(n)$, where we always use $n,m$ to denote the number of nodes and edges of the input graph. For a function $\pi:V\to V'$, we write $\pi^{-1}(v')=\{v\in V\mid \pi(v)=v'\}$. All algorithms in this paper are deterministic.

\paragraph{Graph.}
All graphs of this paper, unless specified, are directed graphs with edge weights, denoted by $G=(V,E,w)$. We assume edge weights are \emph{integers} within the range $[-W,W]$.
We use standard definitions of graph terminologies for directed graphs which can be found in textbooks, including paths, distances, strongly connected components (SCCs), negative cycles, shortest-path trees, etc. For a cycle or path $H$, we use $w(H)$ to denote the sum of weights of edges on the cycle or path, and $|H|$ to denote the number of edges on the cycle or path (note that a path can have repeated edges; in this case, both the weight and the number of such an edge are counted multiple times). We use $w_{\ge 0}$ to denote the weight truncated at $0$, i.e., $w_{\ge 0}(e)=\max\{0,w(e)\}$ for every $e\in E$. Similarly, we write $G_{\ge 0}=(V,E,w_{\ge 0})$.

In many cases of this paper, a path is not a \emph{simple path}, so a path can have repeated vertices (but usually a shortest path is a simple path). We use $\Start(p)$ and $\End(p)$ to denote the starting vertex and ending vertex of a directed path $p$. We use $p_1\oplus p_2$ to denote the concatenation of two paths: if $p_1=(v_1,\dots,v_z)$ and $p_2=(u_1,\dots,u_{z'})$, then $p_1\oplus p_2=(v_1,\dots,v_z,u_1,\dots,u_{z'})$. Notice that $p_1\oplus p_2$ is a path only if $(v_z,u_1)$ is an edge.

We use $V(G)$ and $E(G)$ to denote the vertex set and edge set of a graph $G$. We use $w_G$ to stress that $w_G$ is the weight function for $G$. We write $\dego_G(u)$ to denote the number of edges $(u,v)$ for $v\in V$, and $\degi_G(v)$ to denote the number of edges $(u,v)$ for $u\in V$. We write $\deg_G(v)=\dego_G(v)+\degi_G(v)$. We write $\dego_G(A)=\sum_{v\in A}\dego_G(v)$, $\degi_G(A)=\sum_{v\in A}\degi_G(v)$, and $\deg_G(A)=\sum_{v\in A}\deg_G(v)$.

\paragraph{Graph distance.}
We use $\dist_G(s,t)$ to denote the distance from $s$ to $t$. When $G$ is clear from the context, we simply write $\dist(s,t)$. We write
\[
\Bo_G(s,d)=\{v\in V(G)\mid \dist_G(s,v)\le d\}
\quad\text{and}\quad
\Bi_G(s,d)=\{v\in V(G)\mid \dist_G(v,s)\le d\}.
\]
We allow distance to be $+\infty$ (in the case of not connected) or $-\infty$ (in the case of reaching a negative cycle) so that every distance is well-defined.

Given a directed graph $G=(V,E,w)$ and a vertex set $A\subseteq V$, the \emph{diameter} of $A$ is defined as $\max_{s,t\in A}\dist_{G[A]}(s,t)$ where $G[A]$ is the induced subgraph of $G$ on $A$. The diameter of $A$ is also the diameter of $G[A]$.

\paragraph{The single-source shortest path problem (SSSP).}
Given a directed graph $G=(V,E)$ with integer (possibly negative) edge weights and a source $s\in V$, either output $\dist_G(s,v)$ for every $v\in V$, or output a negative cycle of $G$. It is folklore that we can get a shortest-path tree from a distance function, and vice versa.

\paragraph{Potential adjustment.}
Given a \emph{potential} $\phi:V\to \bbZ$ and a graph $G=(V,E,w)$, we define an adjusted weight $w_\phi(u,v)=w(u,v)+\phi(u)-\phi(v)$. We define the adjusted graph $G_\phi=(V,E,w_\phi)$.

The following lemma shows the power of potential adjustment. The lemma is folklore and was already used in \cite{Johnson77}.

\begin{lemma}[\cite{Johnson77}]\label{lem:potentialequivalence}
Let $G$ be a graph and $\phi:V\to\bbZ$ be a potential function. If $G$ has no negative cycle, a subgraph is a shortest-path tree of $G$ iff it is a shortest-path tree of $G_\phi$. More precisely, let $d^s_G(u)$ be the distance in $G$ from $s$ to $u$, then
\[
d^s_{G_\phi}(u)=d^{s}_{G}(u)+\phi(s)-\phi(u)\, .
\]
\end{lemma}

\newcommand{\PC}{\textsc{PathCover}}

\section{Path Cover via a Clustered Graph}
In this section, we present a method for covering every bounded-length path in a directed graph using a clustered graph.  
We begin with the necessary definitions. All graphs in this section have non-negative edge weights.

\subsection{Definitions and the Main Theorem}

We first formally define graph projections introduced in the introduction.

\begin{definition}[Graph Projections]\label{def:inherited}
    Let $G=(V,E,w)$ be a graph.  
    A \emph{projection onto $G$} is a graph $G'=(V',E',w')$ together with a projection map $\pi=\pi_{G'}$ where $\pi$ is a \emph{weight-preserving graph homomorphism}: for every $(x,y)\in E(G')$, we have $(\pi(x),\pi(y))\in E(G)$ and $w'(x,y)=w(\pi(x),\pi(y))$.
\end{definition}
When we refer to a projection $G'$, we implicitly assume its map $\pi_{G'}$.

\begin{definition}[Path Covering]\label{def:pathcovering}
Given a path $p=(v_1,\dots,v_z)$ in $G$, we say that a projection $(G',\pi)$ onto $G$ \emph{covers} $p$ if there exists a path $p'=(v'_1,\dots,v'_z)$ in $G'$ such that $\pi(v'_i)=v_i$ for every $i\in[z]$. We call $p'$ the \emph{lift} of $p$.
The projection is \emph{$d$-path-covering} if every simple path in $G$ of length at most $d$ is covered.
\end{definition}

Our goal is to cover every simple path of length at most~$d$ by a projection whose SCCs have small diameter.  
For this purpose, we define clustered graphs as follows.

\begin{definition}[Clustered Graphs]
    We say that a directed graph $G'$ is a \emph{$d$-clustered graph} if every strongly connected component (SCC) of $G'$ has diameter at most $d$.  
\end{definition}

\begin{definition}[Path Cover]\label{def:pathcover}
    A \emph{$d$-path cover} of a directed graph $G=(V,E,w)$ with \emph{diameter slack}~$\lambda$ is a $d$-path-covering projection $G'$ where $G'$ is a $\lambda\cdot d$-clustered graph.    
\end{definition}

We will prove the following theorem in this section.

\begin{theorem}[Path Cover via Clustered DAGs]\label{thm:pathcover}
There exists a deterministic algorithm that, given a directed graph $G=(V,E,w)$ with non-negative edge weights, a diameter parameter $d$, and a slack parameter $\lambda \ge 10000\log^4 n$, constructs a $d$-path cover $G'$ of $G$ with diameter slack $\lambda$ and size at most 
\[
\left(1+\tfrac{100\log^2 n}{\sqrt{\lambda}}\right)\cdot m,
\]
in $O(\sqrt{\lambda}\cdot m\log n)$ time.

Moreover, we have
\[
\sum_{v\in V(G')}\deg_G(\pi(v))\le \left(1+\tfrac{100\log^2 n}{\sqrt{\lambda}}\right)\cdot m.
\]
\end{theorem}

To get \Cref{thm:path cover intro}, we only need to set $\lambda=10000\log^6n$.

\begin{remark}
    The running time can be improved to $O(m\log^2 n)$, independent of $\lambda$, if we relax the structure of the path cover so that we only guarantee \emph{weak diameter w.r.t.~$G$}. More precisely, for each strongly connected component $U$ in $G'$, $\max_{u,v \in U}\dist_G(\pi(u),\pi(v))\le d\lambda$. This weaker guarantee is sufficient for our negative-weight shortest path algorithms, and we can improve the running time from $O(m\log^8 n \log(nW))$ to $O(m\log^5 n \log(nW))$ in this way.
    However, we choose to construct the path cover with a strong diameter guarantee to ease the understanding of this new concept.
\end{remark}

\subsection{The Algorithm}\label{subsec:pathcoveralg}

In this section, we describe the algorithm for \Cref{thm:pathcover}. Let $G=(V,E,w)$ be the input directed graph and $d$ be the input diameter parameter.

\newcommand{\led}{\mathrm{rep}}
The following definition is useful when describing algorithms. 
Let $G'$ be a projection onto $G$. We say a vertex $v\in V(G)$ is \emph{present} (in the projection) if $\pi^{-1}(v)\neq\emptyset$.
We say the projection has \emph{representatives} if every present vertex $v\in V(G)$ has a unique representative $\led_{G'}(v)\in V(G')$ with $\pi_{G'}(\led_{G'}(v))=v$. In this case, we say $\led_{G'}(v)$ is a \emph{representative}.\footnote{The representative in our case is basically the first appearance of the copy of $v$ in $G'$. First is defined on the SCC order of $G'$ (we will have the guarantee that each SCC is a subgraph of $G$, which does not have repeated vertices).}

\begin{definition}[Layered Projection]\label{def:layered}
    Let $G_i$ be projections onto $G$ with representatives for $i=1,\dots,z$. 
    A \emph{layered projection} $G' = \mathsf{Layer}((G_1,\dots,G_z)\rightarrow G)$ induced from $(G_1,\dots,G_z)$ onto $G$ with representatives is constructed as follows. First, initialize $G'$ as a disjoint union of $G_1,\dots,G_z$ with the same projection map on vertices\footnote{When $G_i$ and $G_j$ share the same vertex $v$, we make different copies of $v$ in $G'$.}. For every $v\in V(G)$, suppose $v$ is present in $G_i$ where $i$ is the smallest such index, select the representative $\led_{G'}(v)$ to be $\led_{G_i}(v)$.
    Then, for every two vertices $u'\in V(G_i)$ and $v'\in V(G_j)$, if $i<j$, $(\pi_{G_i}(u'),\pi_{G_j}(v'))\in E$ and $v'$ is a representative in $G'$, we add the edge $(u',v')$ to $G'$ with weight $w(\pi_{G_i}(u'),\pi_{G_j}(v'))$.
\end{definition}

The algorithm is recursive. For convenience, we use $\PC(A)$ to denote the subroutine that takes a vertex set $A\subseteq V$ as input and returns a $d$-path cover of $G[A]$ with representatives. It suffices to call $\PC(V)$ to finish \Cref{thm:pathcover}.

We define the following parameters:
\[
\eps=\frac{1}{\sqrt{\lambda}},\qquad \eps'= \frac{9\log n}{\lambda}.
\]

\paragraph{Base case.} When $|A|=1$, we simply return $G[A]$ with the trivial projection map and representative.

\paragraph{Recursive step.} 
Pick an arbitrary vertex $u$ in $A$. We execute the following two procedures in an interleaved fashion: we perform one memory access of each procedure in turn, alternating between them, until one of them terminates. As soon as one procedure stops, the entire process halts. 

\newcommand{\Ii}{i^{\mathrm{in}}}
\newcommand{\Io}{i^{\mathrm{out}}}
\begin{enumerate}
    \item \textbf{(Forward growing)} Run Dijkstra's algorithm to find the smallest $\Io \in \mathbb{N}_{>0}$ such that
    \[
    \deg_G(\Bo_{G[A]}(u, \Io \cdot d)) \le \left(1 + \eps' \right) \cdot \deg_G(\Bo_{G[A]}(u, (\Io-1) \cdot d)).
    \]
    Such an $\Io$ must exist because for sufficiently large $i$, we have $\Bo_{G[A]}(u, i \cdot d) = G[A]$.

    To be precise, Dijkstra's algorithm fixes the distances from $u$ to other nodes in increasing order. We run it only up to the point where the distance exceeds $i \cdot d$, in order to compute $\Bo_{G[A]}(u, i \cdot d)$. This can take sublinear time because the algorithm stops once the desired index $\Io$ is found.

    For technical reasons, we assume that when Dijkstra's algorithm fixes the distance of a node $v$, it performs $\deg_G(v)$ memory accesses, i.e., it looks up all adjacent edges of $v$ in $G$ (but only performs $\dego_{G[A]}(v)$ relaxations). This setting is purely for the convenience of analysis.
    
    \item \textbf{(Backward growing)} Identical to the forward growing procedure, but with all occurrences of `$\mathrm{out}$' replaced by `$\mathrm{in}$'.
\end{enumerate}

For simplicity, we use the following notation. Notice that the algorithm only knows $\iBo,\bBo$ but does not know $\iBi,\bBi$ if forward growing stops earlier than backward growing, and vice versa: 
\begin{align*}
\iBo &= \Bo_{G[A]}(u,\Io\cdot d), &
\bBo &= A - \Bo_{G[A]}(u,(\Io-1)\cdot d),\\
\iBi &= \Bi_{G[A]}(u,\Ii\cdot d), &
\bBi &= A - \Bi_{G[A]}(u,(\Ii-1)\cdot d).
\end{align*}

\newcommand{\To}{T^{\mathrm{out}}}
\newcommand{\tTo}{\tilde{T}^{\mathrm{out}}}
\newcommand{\Ti}{T^{\mathrm{in}}}
\newcommand{\tTi}{\tilde{T}^{\mathrm{in}}}
Without loss of generality, assume forward growing stops earlier than backward growing (the other case is handled symmetrically). Consider two cases.
\begin{enumerate}
	\item \textbf{Case 1: $\deg_G(\iBo)< \left(1-\eps\right)\deg_G(A)$.} In this case, let 
    \[
    \Ho\leftarrow \PC\!\left(\iBo\right),\qquad
    \bHo\leftarrow \PC\!\left(\bBo\right).
    \]
    Let $G'$ be the layered projection induced from $(\bHo,\Ho)$ (\Cref{def:layered}). Return $G'$ as the output of $\PC(A)$. 
    
	\item \textbf{Case 2: $\deg_G(\iBo)\ge (1-\eps)\deg_G(A)$.} Continue the backward growing until it stops. Now the algorithm also knows $\iBi,\bBi$.
    Let $\To$ be the shortest-path tree in $G[\iBo]$ with root $u$, and let $\Ti$ be the reversed shortest-path tree in $G[\iBi]$ with root $u$ (i.e., $\Ti$ contains the shortest paths from every node in $\iBi$ to $u$).

    Let 
    \[
    M\leftarrow \iBo\cap\iBi.
    \]
    Let $\tTo$ be the subtree of $\To$ that contains all nodes in $V(\To)$ that can reach $M$ in $\To$.
     Similarly, define $\tTi$ to be the subtree of $\Ti$ that contains all nodes in $V(\Ti)$ that can be reached from $M$ in $\Ti$. 
    
    Define
	\[
 \Hm\leftarrow G[V(\tTo)\cup V(\tTi)].
 \]
    Notice that $\Hm$ is a subgraph of $G$, so a projection map and representatives can be defined trivially.
    
    Make the following recursive calls:
    \[
    \bHi\leftarrow \PC\!\left(\bBi\right),\qquad
    \tHo\leftarrow \PC\!\left(\bBo\cap \iBi \right).
    \]
    Let $G'$ be the layered projection induced from $(\tHo,\Hm,\bHi)$. Return $G'$ as the output of $\PC(A)$.
\end{enumerate}

For a graph illustration, see \Cref{fig:example} which shows Case 1. Case 2 can be viewed as after applying Case 1 to get $(\Hi,\bHi)$, we recursively build $\Hi$ by $(\tHo,\Hm)$. 
\subsection{Correctness}\label{subsec:pathcovercorrectness}

We prove the correctness of the algorithm described in \Cref{subsec:pathcoveralg}. Specifically, $\PC(A)$ returns a $d$-path cover of $G[A]$. We proceed by induction on $|A|$.

\paragraph{Base case.} Suppose $|A|=1$. The only SCC has diameter $0$. Thus, $G[A]$ is a $d$-clustered graph. Moreover, the trivial projection map certifies that all paths are covered.

\paragraph{Induction step.} Assume $\PC$ is correct for every input size smaller than $|A|$. We first show the easy part: $G'$ is a $\lambda\cdot d$-clustered graph. We need the following lemma.

\begin{lemma}\label{lem:sizeofi}
    Both $\Io$ and $\Ii$ are at most $\lambda/4$.
\end{lemma}
\begin{proof}
    For every $i<\Io$, by the definition of $\Io$,
\[
    \deg_G(\Bo_{G[A]}(u, i \cdot d)) > \left(1 + \eps' \right) \cdot \deg_G(\Bo_{G[A]}(u, (i-1) \cdot d)).
\]
Combining these inequalities for every $i<\Io$ gives
\[
m\ge \deg_G(\Bo_{G[A]}(u, \Io \cdot d)) >\left(1 + \eps' \right)^{\Io-1}
=\left(1+\frac{9\log n}{\lambda}\right)^{\Io-1}.
\]
Since $\lambda\ge 10000\log^4 n$, we obtain $\Io\le \lambda/4$. The same proof holds symmetrically for $\Ii$. 
\end{proof}

\paragraph{Bounding diameter.}
To show $G'$ is a $\lambda\cdot d$-clustered graph, let $C$ be an SCC of $G'$; we must show $\mathrm{diam}(C)\le \lambda\cdot d$. If the algorithm goes to Case~1, then $G'$ is the layered projection induced from $(\bHo,\Ho)$. By \Cref{def:layered}, there are only edges from $\bHo$ to $\Ho$ but not backwards, so $C$ is either completely inside $\bHo$ or inside $\Ho$; in either case, $C$ has diameter at most $\lambda\cdot d$ by the induction hypothesis. If the algorithm goes to Case~2, then $G'$ is the layered projection induced from $(\tHo,\Hm,\bHi)$. Similarly, if $C$ is completely inside $\tHo$ or $\bHi$, then $C$ has diameter at most $\lambda\cdot d$ by the induction hypothesis. If $C$ is inside $\Hm$, we show that $\Hm$ is strongly connected and has diameter at most $\lambda\cdot d$. Recall
\[
\Hm\leftarrow G[V(\tTo)\cup V(\tTi)],
\]
where $\tTo$ is the subtree of $\To$ whose every node can reach $W=\iBo\cap \iBi$, and $\tTi$ is the subtree of $\Ti$ whose every node can be reached from $W$.
Recall also
\[
\iBo = \Bo_{G[A]}(u,\Io\cdot d),\qquad
\iBi = \Bi_{G[A]}(u,\Ii\cdot d),
\]
and $\To,\Ti$ are shortest-path trees in $G[\iBo]$ and $G[\iBi]$. By \Cref{lem:sizeofi}, each shortest-path tree has depth at most $(\lambda/4)\cdot d$. Thus, every node in $\tTo$ can reach $W$ with distance at most $(\lambda/4)\cdot d$ in $\tTo$; every node in $\tTi$ can be reached from $W$ with distance at most $(\lambda/4)\cdot d$ in $\tTi$. Similarly, $u$ can reach every node in $\tTo$ within distance at most $(\lambda/4)\cdot d$ in $\tTo$, and $u$ can be reached from every node in $\tTi$ within distance at most $(\lambda/4)\cdot d$ in $\tTi$. We conclude that the diameter of $\Hm$ is at most $\lambda\cdot d$.

\paragraph{Covering all length-$d$ paths.}
Next we show that $G'$ is a $d$-path cover of $G[A]$. Let $p$ be an arbitrary path of $G[A]$ with length at most $d$. We prove that $G'$ covers $p$ (see \Cref{def:pathcovering}). Moreover, we prove a stronger statement: there exists a lift $p'$ of $p$ such that $\Start(p')=\led(\Start(p))$.\thatchaphol{This is the important property that will give a lot of intuition/motivation why we define representation. It should be highlight much earlier (not buried here). It should be in both the overview and after you introduce the notion of representation in the body.}

Since $G'$ depends on whether the algorithm goes to Case~1 or Case~2, we prove the two cases separately.

\paragraph{Induction step (Case 1).} Here $G'$ is the layered projection induced from $(\bHo,\Ho)$ where 
\begin{itemize}
    \item $\bHo$ is a path cover of $G[\bBo]$, and
    \item $\Ho$ is a path cover of $G[\iBo]$.
\end{itemize}

If $p$ is completely contained in $G[\bBo]$, then by the induction hypothesis, $p$ is covered by $G'$. Moreover, the lift $p'$ of $p$ has $\Start(p')=\led_{\bHo}(\Start(p))$. By \Cref{def:layered}, $\led_{\bHo}(\Start(p))=\led_{G'}(\Start(p))$ since $\bHo$ precedes $\Ho$. Thus, $\Start(p')=\led_{G'}(\Start(p))$.

Now assume $p$ is not completely contained in $G[\bBo]$, so there exists a vertex on $p$ that lies in $A-\bBo$. Let $a$ be the first such vertex on $p$. 

Let $p_1$ be the subpath of $p$ before vertex $a$ (excluding $a$; possibly empty), and let $p_2$ be the subpath of $p$ after vertex $a$ (including $a$). 
By definition, $p_1$ is contained in $G[\bBo]$, so by induction, $p_1$ is covered by $\bHo$. 

We now show that $p_2$ is contained in $G[\iBo]$. Recall
\[
A-\bBo = \Bo_{G[A]}(u,(\Io-1)\cdot d),\qquad
\iBo = \Bo_{G[A]}(u,\Io\cdot d). 
\]
Thus, if $a\in A-\bBo$, any path of length at most $d$ starting at $a$ is contained in $\iBo$.

Hence $p_1$ is covered by $\bHo$ and $p_2$ is covered by $\Ho$, with lifts $p'_1,p'_2$ respectively. We have $\pi(\End(p'_1))=\End(p_1)$ and $\Start(p'_2)=\led_{\Ho}(\Start(p_2))$. 
Since $\Start(p_2)=a\in A-\bBo$, $a$ is not present in $\bHo$, so $\led_{G'}(a)=\led_{\Ho}(a)$ and $\Start(p'_2)=\led_{G'}(a)$. 
By \Cref{def:layered}, there is an edge $(\End(p'_1),\Start(p'_2))$ in $G'$, so $p'_1\oplus p'_2$ is a path in $G'$ and is a lift of $p=p_1\oplus p_2$.

If $p'_1$ is non-empty, then $\Start(p'_1)=\led_{\bHo}(\Start(p_1))$ by induction, hence $\Start(p'_1\oplus p'_2)=\led_{G'}(\Start(p))$ since $\bHo$ precedes $\Ho$ in $G'$. Otherwise, $p'_1$ is empty, so $p$ starts at $a\in V-\bBo$. In this case, since $a$ is not present in $\bHo$, we must have $\led_{G'}(a)=\led_{\Ho}(a)=\Start(p'_2)$, where the second equality is by induction.

\paragraph{Induction step (Case 2).} Here $G'$ is the layered projection induced from $(\tHo,\Hm,\bHi)$, where 
\begin{itemize}
    \item $\tHo$ is a path cover of $G[\bBo\cap\iBi]$,
    \item $\Hm\supseteq  G[\iBo\cap\iBi]$,
    \item $\bHi$ is a path cover of $G[\bBi]$.
\end{itemize}

If $p$ is completely contained in $G[\bBo\cap\iBi]$, then $p$ is covered by $G'$. Moreover, since $\tHo$ precedes $\Hm,\bHi$, we have $\led_{G'}(\Start(p))=\led_{\tHo}(\Start(p))$, which equals the starting vertex of the lift of $p$ in $G'$.

Otherwise, there is a vertex of $p$ in $A-(\bBo\cap\iBi)$; let $a$ be the first such vertex. Let $p_1$ be the subpath of $p$ before $a$ (excluding $a$; possibly empty), and let $p_2$ be the subpath of $p$ starting at $a$. Then $p_1\subseteq G[\bBo\cap\iBi]$, so $p_1$ is covered by $\tHo$ with lift $p'_1$. If $p_1$ is non-empty, then $\led_{G'}(\Start(p))=\led_{\tHo}(\Start(p))=\Start(p'_1)$ by induction.

We first show that $p_2$ is completely contained in $V(\Hm)\cup\bBi$. Suppose Recall
\begin{align*}
A-(\bBo\cap\iBi)&\subseteq (A-\bBo)\cup (A-\iBi)=\Bo_{G[A]}(u,(\Io-1)\cdot d)\cup (A-\Bi_{G[A]}(u,\Ii\cdot d)),\\
V(\Hm)\cup \bBi&\supseteq (\iBo\cap \iBi)\cup \bBi\supseteq \Bo_{G[A]}(u,\Io\cdot d)\cup(A-\Bi_{G[A]}(u,(\Ii-1)\cdot d)).
\end{align*}
Hence, if $a\in A-(\bBo\cap\iBi)$ and $p_2$ has length at most $d$, it is contained in $V(\Hm)\cup \bBi$. We now argue within $V(\Hm)\cup \bBi$.

If $p_2$ is completely contained in $V(\Hm)$, then $p_2$ is covered by $\Hm$ trivially (as $\Hm$ is an induced subgraph of $G$) with a lift $p'_2$. Moreover, $\led_{G'}(a)=\led_{\Hm}(a)=a=\Start(p'_2)$ since $a$ is not present in $\tHo$.

Otherwise, $p_2$ has a vertex in $A-(\bBo\cap\iBi)-V(\Hm)$; let $b$ be the first such vertex. Let $p_{2,1}$ be the subpath of $p_2$ before $b$ (excluding $b$; possibly empty) and let $p_{2,2}$ be the subpath of $p_2$ from $b$ onward. Then $p_{2,1}\subseteq V(\Hm)$ and hence is covered by $\Hm$ with lift $p'_{2,1}$. If $p_{2,1}$ is non-empty, then as above $\led_{G'}(a)=\Start(p'_{2,1})$.

We prove that $p_{2,2}\subseteq G[\bBi]$, so that it is covered by $\bHi$ with lift $p'_{2,2}$. Note
\begin{align*}
A-(\bBo\cap\iBi)-V(\Hm)
&\subseteq A-(\bBo\cap\iBi)-(\iBo\cap\iBi) 
= A-\iBi
= A-\Bi_{G[A]}(u,\Ii\cdot d),\\
\bBi &= A-\Bi_{G[A]}(u,(\Ii-1)\cdot d).
\end{align*}
Since $b=\Start(p_{2,2})\in A-\Bi_{G[A]}(u,\Ii\cdot d)$ and $p_{2,2}$ has length at most $d$, $p_{2,2}$ is contained in $\bBi$ by definition (otherwise it would enter $\Bi_{G[A]}(u,(\Ii-1)\cdot d)$, contradicting $b\in A-\Bi_{G[A]}(u,\Ii\cdot d)$). 

Moreover, $b$ is not present in $\tHo$ or $\Hm$, so $\led_{G'}(b)=\led_{\bHi}(b)=\Start(p'_{2,2})$ by induction. 

Now $p_1$ is covered by $\tHo$ with lift $p'_1$, $p_{2,1}$ is covered by $\Hm$ with lift $p'_{2,1}$, and $p_{2,2}$ is covered by $\bHi$ with lift $p'_{2,2}$. Assuming all are non-empty (the empty cases are analogous), we have $\pi(\End(p'_1))=\End(p_1)$, $\Start(p'_{2,1})=\led_{G'}(a)$ and $\pi(\End(p'_{2,1}))=\End(p_{2,1})$, and $\Start(p'_{2,2})=\led_{G'}(b)$. By \Cref{def:layered}, edges $(\End(p'_1),\Start(p'_{2,1}))$ and $(\End({p'_{2,1}}),\Start(p'_{2,2}))$ are in $G'$, so $p'_1\oplus p'_{2,1}\oplus p'_{2,2}$ is a path in $G'$ and is a lift of $p=p_1\oplus p_{2,1}\oplus p_{2,2}$. 

\subsection{Size and Time Complexity}

\newcommand{\Mo}{m^{\mathrm{out}}}
\newcommand{\bMo}{\bar{m}^{\mathrm{out}}}
\newcommand{\Mi}{m^{\mathrm{in}}}
\newcommand{\bMi}{\bar{m}^{\mathrm{in}}}

The algorithm is recursive. We first prove the following lemma.
\begin{lemma}\label{lem:projection size}
    Suppose $G'\leftarrow \PC(A)$. The size of $G'$ is at most $\sum_{v\in V(G')}\deg_{G[A]}(\pi(v))$.
\end{lemma}
\begin{proof}
    By induction on $|A|$. When $|A|=1$, there are no edges in $G'$, so the lemma is trivial. Suppose $|A|>1$. The algorithm returns different $G'$ based on two cases.

    \paragraph{Case 1.} The size of $G'$ contains two parts: (1) the sizes of $\bHo,\Ho$, which, by induction, are at most 
    \[
    \sum_{v\in V(\bHo)}\deg_{G[\bBo]}(\pi(v))+\sum_{v\in V(\Ho)}\deg_{G[\iBo]}(\pi(v)).
    \]
    (2) The edges added by \Cref{def:layered}. For every $v\in V(\bHo)$, the number of edges added from $v$ equals $\deg_{G[A]}(\pi(v))-\deg_{G[\bBo]}(\pi(v))$, because only vertices in $A-\bBo$ can have representatives in $\Ho$; all other vertices are present in $\bHo$. Summing over $\sum_{v\in V(\bHo)}(\deg_{G[A]}(\pi(v))-\deg_{G[\bBo]}(\pi(v)))$ and $\sum_{v\in V(\bHo)}\deg_{G[\bBo]}(\pi(v))+\sum_{v\in V(\Ho)}\deg_{G[\iBo]}(\pi(v))$ gives the bound $\sum_{v\in V(G')}\deg_{G[A]}(\pi(v))$.

\paragraph{Case 2.}

\smallskip\noindent\textit{(i) The \(\tHo\)-term.}
All edges with both endpoints in \(\tHo\) lie inside the induced subgraph
\(G[\bBo\cap\iBi]\). Thus
\[
|E(\tHo)| \;\le\; \sum_{v\in V(\tHo)} \deg_{G[\bBo\cap \iBi]}(\pi(v)).
\]
Any additional edge incident to \(v\in V(\tHo)\) that goes to
\(\Hm\cup\bHi\) must use a representative outside \(\bBo\cap\iBi\), so the
number of such outgoing edges is at most
\[
\sum_{v\in V(\tHo)}
\Bigl(\deg_{G[A]}(\pi(v))-\deg_{G[\bBo\cap \iBi]}(\pi(v))\Bigr).
\]
Consequently, the total number of edges charged to \(\tHo\) is at most
\[
\sum_{v\in V(\tHo)} \deg_{G[A]}(\pi(v)).
\]

\smallskip\noindent\textit{(ii) The \(\Hm\)-term.}
Edges internal to \(\Hm\) are contained in the induced subgraph on
\(V(\Hm)\), so
\[
|E(\Hm)| \;\le\; \sum_{v\in V(\Hm)} \deg_{G[V(\Hm)]}(\pi(v)).
\]
Edges from \(\Hm\) to \(\bHi\) are added only from the \(\Hm\)-side, and for
each \(v\in V(\Hm)\) there are at most
\[
\deg_{G[A]}(\pi(v))-\deg_{G[V(\Hm)]}(\pi(v))
\]
such edges. Hence the total number of edges charged to \(\Hm\) is at most
\[
\sum_{v\in V(\Hm)} \deg_{G[A]}(\pi(v)).
\]

\smallskip\noindent\textit{(iii) The \(\bHi\)-term.}
By the induction hypothesis applied within \(\bHi\), the edges with both
endpoints in \(\bHi\) are at most
\[
|E(\bHi)| \;\le\; \sum_{v\in V(\bHi)} \deg_{G[A]}(\pi(v)).
\]
Inter-piece edges incident to \(\bHi\) are \emph{not} charged here (they were
already charged to their tails in \(\tHo\) or \(\Hm\)), so there is no
double counting.

\smallskip
Summing the contributions from (i)–(iii) yields
\[
|E(G')|
\;\le\;
\sum_{v\in V(\tHo)} \deg_{G[A]}(\pi(v))
+
\sum_{v\in V(\Hm)} \deg_{G[A]}(\pi(v))
+
\sum_{v\in V(\bHi)} \deg_{G[A]}(\pi(v))
\]
\[=
\sum_{v\in V(G')} \deg_{G[A]}(\pi(v)),
\]
which completes the proof.

\end{proof}

Thus, it remains to bound $\sum_{v\in V(G')}\deg_G(\pi(v))$, which we prove next.

\begin{lemma}\label{lem:degsize}
    Suppose $G'\leftarrow \PC(A)$. Then
    \[
    \sum_{v\in V(G')}\deg_G(\pi(v))\le \left(1+9(\log (\deg_G(A)))\cdot \frac{\eps'}{\eps}\right)\deg_G(A).
    \]
    The running time of $\PC(A)$ is at most $\log (\deg_G(A))\cdot O(\deg_G(A)/\eps)$.
\end{lemma}
\begin{proof}
    By induction on $|A|$. For $|A|=1$, $\PC$ returns $G[A]$, and the claim is trivial.

    The algorithm first runs two procedures (forward and backward growing) in lockstep, one edge access at a time, and stops as soon as one stops. The time between two edge accesses is at most $O(\log n)$ using a standard heap. Assume forward growing stops earlier.

    The forward procedure runs Dijkstra from $u$ until it finds $\Io$, costing $O(\log n\cdot \deg_G(\iBo))$.

    We analyze the two cases.

    \paragraph{Case 1: $\deg_G(\iBo)< (1-\eps)\deg_G(A)$.}
    Here $V(G')=V(\Ho)\cup V(\bHo)$ with $\Ho=\PC(\iBo)$ and $\bHo=\PC(\bBo)$. By induction,
    \begin{align*}
    \sum_{v\in V(G')}\deg_G(\pi(v))
    &\le \left(1+\frac{9\eps'\log\bigl|\deg_G(\iBo)\bigr|}{\eps}\right)\deg_G(\iBo)
    +\left(1+\frac{9\eps'\log\bigl|\deg_G(\bBo)\bigr|}{\eps}\right)\deg_G(\bBo)\\
    &\le \left(1+\frac{9\eps'(\log (1-\eps)+\log\bigl|\deg_G(A)\bigr|)}{\eps}\right)
    (1+\eps')\deg_G(A-\bBo)\\
    &\quad+\left(1+\frac{9\eps'\log\bigl|\deg_G(\bBo)\bigr|}{\eps}\right)\deg_G(\bBo)\\
    &\le \left(1+\frac{9\eps'\log\bigl|\deg_G(A)\bigr|}{\eps}\right)\deg_G(A).
    \end{align*}
    The second inequality uses $\deg_G(\iBo)< (1-\eps)\deg_G(A)$, $\lambda\ge 1000\log^2 n$, and the definitions of $\iBo,\bBo$. The third inequality uses 
        \[
        \left(1+\frac{9\eps'(\log(1-\eps)+\log\bigl|\deg_G(A)\bigr|)}{\eps}\right)\le 2,
        \qquad
        \frac{\log(1-\eps)}{\eps}\le -0.4,
        \]
    after plugging in $\eps',\eps$.

    The running time consists of two recursive calls and building the layered projection. The latter is bounded by $\sum_{v\in V(G')}\deg_G(\pi(v))$, since every edge we add corresponds to some $v'\in V(G')$ and an edge $(\pi(v'),v)$ in $G$. The total complexity is
    \[
    O(\log n\cdot \deg_G(\iBo))+\tfrac{\log (\deg_G(A))+\log(1-\eps)}{\eps}\cdot O(\log n\cdot (\deg_G(\iBo)+\deg_G(\bBo)))+O(\deg_G(A))
    \]
    \[
    \le \log (\deg_G(A))\cdot O(\log n\cdot \deg_G(A)/\eps),
    \]
    where the first term is for finding $\Io$ (or $\Ii$), the second for recursion by induction, and the last for constructing $G'$, which is $\le 2\deg_G(A)$ since $\lambda>10000\log^4 n$.
    
    \paragraph{Case 2: $\deg_G(\iBo)\ge (1-\eps)\deg_G(A)$.}
    Now $V(G')=V(\tHo)\cup V(\Hm)\cup V(\bHi)$ with $\tHo=\PC(\bBo\cap \iBi)$, $\Hm=G[\iBo\cap \iBi]$, and $\bHi=\PC(\bBi)$. We need:

    \begin{lemma}\label{lem:largemid}
        $\deg_G(\bBo),\deg_G(\bBi)\le 2\eps\deg_G(A)$.
    \end{lemma}
    \begin{proof}
        Since forward growing stops earlier than backward growing, the backward procedure accesses at least as many edges before stopping. As each accessed edge contributes to $\deg_G(\iBi)$, we have $\deg_G(\iBo),\deg_G(\iBi)\ge (1-\eps)\deg_G(A)$. By the definitions of $\bBo,\bBi$,
        \[
        \deg_G(\bBo),\deg_G(\bBi)\le \eps\deg_G(A)+\eps'\deg_G(A)\le 2\eps\deg_G(A).
        \]
    \end{proof}

    By induction,
\begin{align*}
\sum_{v\in V(G')}\deg_G(\pi(v))
&\le\left(1+\frac{9\eps' \log\bigl|\deg_G(\bBo\cap\iBi)\bigr|}{\eps}\right)\deg_G(\bBo\cap\iBi)\\
&\quad+\left(1+\frac{9\eps'\log\bigl|\deg_G(\bBi)\bigr|}{\eps}\right)\deg_G(\bBi)\\
&\quad+\deg_G(A)\\
&\le\left(1+\frac{9\eps'\,(\log\bigl|\deg_G(A)\bigr|+\log(2\eps))}{\eps}\right)\cdot 2\eps\deg_G(A)\\
&\quad+\left(1+\frac{9\eps'\,(\log\bigl|\deg_G(A)\bigr|+\log(2\eps))}{\eps}\right)\cdot 2\eps\deg_G(A)+\deg_G(A)\\
&\le \deg_G(A)+9\eps\deg_G(A)
\le\left(1+\frac{9\eps' \log\bigl|\deg_G(A)\bigr|}{\eps}\right)\deg_G(A),
\end{align*}
using \Cref{lem:largemid} and the definitions of $\iBo,\bBo$, then plugging in $\eps,\eps'$.

The running time here includes two global SSSP computations costing $O(\deg_G(A)\log n)$, plus recursion time bounded by $(\log(1-\eps)+\log (\deg_G(A)))\cdot O(\log n\cdot \deg_G(A)/\eps)$. Thus, the total is $\log (\deg_G(A))\cdot O(\log n\cdot \deg_G(A)/\eps)$.
\end{proof}

\section{Deterministic Negative Weight SSSP}

In \Cref{subsec:restrictedalg}, we show how to solve the \emph{restricted SSSP} problem (defined below). Bringmann, Cassis, and Fischer \cite{BringmannCF23} showed a black-box reduction from solving SSSP on general graphs to restricted SSSP. Although the main statements in \cite{BringmannCF23} are randomized (see Sections 4 and 5 in \cite{BringmannCF23}), the reduction to a restricted graph is deterministic.

\begin{definition}[\cite{BringmannCF23}]\label{def:restrictedgraphs}
    A \emph{restricted graph} is a directed graph $G=(V,E,w)$ satisfying 
    \begin{itemize}
        \item the edge weights lie in the set $\{-1,0,1,\dots,n\}$,\footnote{In the original definition of \cite{BringmannCF23}, there is no upper bound $n$ on the edge weights. We can assume an upper bound $n$ because any edge of weight $>n$ cannot appear in the shortest-path tree from $s$: otherwise that path would have weight at least $1$ (the most negative edge has weight at least $-1$ and a simple path has at most $n-1$ edges), contradicting that $s$ has a $0$-weight edge to every vertex. This assumption also appears in \cite{FischerHLRS25}.}
        \item for every cycle $C$, we have $w(C)/|C|\ge 1$,
        \item there is a vertex $s\in V$ with an edge of weight $0$ from $s$ to every other vertex.
    \end{itemize}  
    The problem \emph{restricted SSSP} is: given a restricted graph, find a shortest-path tree from $s$.
\end{definition}

Sections 4 and 5 of \cite{BringmannCF23} prove the following reduction.

\begin{lemma}[Sections 4 and 5 of \cite{BringmannCF23}]\label{lem:reducetorestricted}
    If there is a deterministic algorithm for restricted SSSP with running time $T_{\mathrm{RSSSP}}(m,n)$, then there is a deterministic algorithm for general SSSP with running time $O\!\left(T_{\mathrm{RSSSP}}(m,n)\cdot \log (Wn)\right)$.
\end{lemma}

In \Cref{subsec:restrictedalg}, we prove the next lemma.

\begin{lemma}\label{lem:RSSSP}
    There is a deterministic algorithm solving restricted SSSP in $O(m\log^8 n)$ time.
\end{lemma}

Combining \Cref{lem:reducetorestricted,lem:RSSSP} yields \Cref{thm:main}.

\subsection{An Algorithm for Restricted SSSP}\label{subsec:restrictedalg}

We now describe the algorithm for restricted SSSP. The input is a directed graph $G=(V,E,w)$ and a source $s$ satisfying \Cref{def:restrictedgraphs}. The algorithm outputs a shortest-path tree of $G$ from $s$.

Before stating the algorithm, we recall two standard subroutines: (1) computing distances when all negative edges must cross different SCCs, and (2) computing distances when the number of negative edges on every shortest path is small (by alternating Dijkstra and Bellman–Ford). We defer proofs to \Cref{sec:basicSSSP}.

\begin{lemma}\label{lem:DAGsssp}
    Given a directed graph $G=(V,E,w)$ where $w(u,v)\ge0$ for every edge $(u,v)$ with $u,v$ in the same SCC of $G$, there is a deterministic $O(m\log n)$-time algorithm for SSSP on $G$. 
\end{lemma}

\begin{lemma}[SSSP with Few Negative Edges]\label{lem:SSSPfewneg}
    There is a deterministic algorithm which, given a directed graph $G=(V,E,w)$ with no negative cycle, a source $s$, and an integer $k$ such that every $s$-to-$v$ shortest path uses at most $k$ negative edges, computes $\dist_G(s,\cdot)$ in $O(km\log n)$ time.
\end{lemma}

\newcommand{\SSSPFN}{\textsc{kSSSP}}

Our algorithm is recursive. We describe an algorithm $\SSSPFN(H,s,k)$ that solves a restricted SSSP instance $(H,s)$ under the promise that every shortest $s$-to-$v$ path in $H$ uses at most $k$ negative edges. It suffices to call $\SSSPFN(H,s,n)$ to prove \Cref{lem:RSSSP}. 

\paragraph{Base case.} When $k=O(\log^6 n)$, apply \Cref{lem:SSSPfewneg} to $(H,s,k)$ to obtain all distances from $s$ in $H$. For the remainder, assume $k=\omega(\log^6 n)$.

\newcommand{\dcov}{d_{\mathrm{cov}}}
\newcommand{\ddiam}{d_{\mathrm{diam}}}

\paragraph{Building a path cover.} Compute a $\dcov$-path cover $G'_{\ge 0}$ of $H_{\ge 0}$ with at most $(1+1/\log n)\cdot |E(H)|$ edges and diameter slack $\lambda$ such that 
\[
\ddiam\le \lambda\cdot \dcov \;=\; k/2.
\]
By \Cref{thm:pathcover}, it suffices to set
\[
\lambda=10000\log^6 n
\qquad\text{and}\qquad 
\dcov=\frac{k}{20000\log^6 n},
\]
and invoke \Cref{thm:pathcover} to obtain such a $G'_{\ge 0}$.

\paragraph{Recursive call.}
From $G'_{\ge 0}$, form $G'$ by restoring each $0$-weight edge back to its original weight in $H$: for every $(u',v')\in E(G'_{\ge 0})$ with $w(\pi(u'),\pi(v'))<0$, set $w_{G'}(u',v'):=w(\pi(u'),\pi(v'))$.

Define $\GC$ from $G'$ by:
\begin{enumerate}
    \item computing the SCC decomposition of $G'$ and deleting all edges whose endpoints lie in different SCCs, and
    \item adding a super-source $s'$ and an edge of weight $0$ from $s'$ to every vertex of $G'$.
\end{enumerate}
Call $\SSSPFN(\GC,s',k/2)$ to obtain $d_{s'}(\cdot)=\dist_{\GC}(s',\cdot)$.

\paragraph{Computing distances for $H$.}
Construct $G''$ as follows.
\begin{enumerate}
  \item Make $x=2\lambda$ copies of $G'$, denoted $G'_1,\dots,G'_x$. For $u'\in V(G')$, let $u'_i\in V(G'_i)$ be its copy.
  \item For every $1 \le i < x$, each $(u,v)\in E(H)$, and each copy $u'_i \in \pi^{-1}(u)$ in $G'_i$ and the representative $v'_{i+1} \in \pi^{-1}(v)$ of $v$ in $G'_{i+1}$, add the edge $(u'_i,v'_{i+1})$ with weight $w(u,v)$.
  \item Add a source $s''$ and, for every vertex $u'_i$ with $\pi(u')=s$, add the $0$-weight edge $(s'',u'_i)$.
\end{enumerate}

Define a potential $\phi$ on $G''$ by
\[
\phi(u'_i)=d_{s'}(u')\quad\text{for }u'_i\neq s'',\qquad \phi(s'')=0.
\]
We will show that $G''_{\phi}$ satisfies the premise of \Cref{lem:DAGsssp}. Invoke \Cref{lem:DAGsssp} on $G''_{\phi}$ to compute $d_{s''}:V(G'')\to\bbZ$. For every $u\in V(H)$, output
\[
d_s(u)=\min\{\,d_{s''}(u'_i)+\phi(u'_i)\;\mid\; u'_i\in V(G''),\ \pi(u')=u\,\}.
\]

\subsection{Correctness}\label{subsec:RSSSPcorrectness}

\begin{lemma}
    $\SSSPFN(H,s,k)$ correctly outputs $\dist_H(s,\cdot)$ provided that
    \begin{itemize}
        \item $(H,s)$ is a restricted SSSP instance, and
        \item every shortest path from $s$ in $H$ uses at most $k$ negative edges.
    \end{itemize}
\end{lemma}

We proceed by induction on $k$. The base case $k=O(\log^6 n)$ follows directly from \Cref{lem:SSSPfewneg}. For the induction step, assume $k=\omega(\log^6 n)$.

\paragraph{Validity of oracle calls.} We first verify the two oracle calls: the recursive call $\SSSPFN(\GC,s',k/2)$ and the call to \Cref{lem:DAGsssp}.

\begin{lemma}
    $(\GC,s')$ is a restricted SSSP instance.
\end{lemma}
\begin{proof}
    We verify the three bullets in \Cref{def:restrictedgraphs}. Since $G'_{\ge 0}$ is a projection onto $H_{\ge 0}$ (\Cref{def:inherited}) and $G'$ restores weights only by replacing some $0$-weights with the original weights in $H$, $G'$ is a projection onto $H$. As $\GC-\{s'\}$ is a subgraph of $G'$, it is also a projection onto $H$; hence all edge weights lie in $\{-1,0,1,\dots,n\}$ (edges adjacent to $s'$ have weight $0$). 

    Let $C$ be any cycle in $\GC$. Since $s'$ has no incoming edges, $C$ lies in $\GC-\{s'\}$, which projects to a (possibly non-simple) cycle $\pi(C)$ in $H$. By \Cref{def:restrictedgraphs}, $w(\pi(C))/|\pi(C)|\ge 1$, hence $w(C)/|C|\ge 1$. The last bullet holds by construction of $s'$.
\end{proof}

\begin{lemma}
    Every shortest path from $s'$ in $\GC$ uses at most $k/2$ negative edges.
\end{lemma}
\begin{proof}
    Suppose, for contradiction, that a shortest $s'$-to-$v'$ path $p$ in $\GC$ uses $>k/2$ negative edges. Let $p'$ be $p$ with its first vertex (the source $s'$) removed; then $p'$ still uses $>k/2$ negative edges and lies inside a single SCC of $G'$ by construction of $\GC$. Also $w_{\GC}(p')\le 0$, as otherwise the direct $0$-edge from $s'$ to $\End(p')$ would give a strictly shorter path, contradicting minimality of $p$.

    Write $\Start(p')=u'$ and $\End(p')=v'$. Since each SCC of $G'_{\ge 0}$ has (strong) diameter at most $\ddiam\le k/2$, there exists in $H_{\ge 0}$ a path $\tilde{p}$ from $\pi(v')$ to $\pi(u')$ of length at most $k/2$. The walk $\pi(p')\oplus\tilde{p}$ forms a (possibly non-simple) cycle in $H$, so by \Cref{def:restrictedgraphs},
    \[
      w(\pi(p')) + w(\tilde{p}) \;\ge\; |\pi(p')|+|\tilde{p}|.
    \]
    Using $w(\pi(p'))=w_{\GC}(p')\le 0$ and $w(\tilde{p})\le w_{\ge 0}(\tilde{p})\le k/2$, the number of \emph{all} edges in $p'$ is
    \[
      |p'|  = |\pi(p')|  \;\le\; k/2,
    \]
    contradicting that $p'$ has more than $k/2$ \emph{negative} edges.
\end{proof}

\begin{lemma}\label{lem:validDAGcall}
    For every edge $(u'',v'')\in E(G''_\phi)$ with $u'',v''$ in the same SCC of $G''_\phi$, we have $w_{G''_\phi}(u'',v'')\ge 0$. 
\end{lemma}
\begin{proof}
    The source $s''$ has no incoming edges, so assume $u'',v''\neq s''$. By construction of $G''$, edges go only from $G'_i$ to $G'_j$ with $j\ge i$, hence if $u'',v''$ lie in the same SCC of $G''_\phi$, they must belong to the same copy $G'_i$ and correspond to $u',v'\in V(G')$ lying in the same SCC of $G'$. It suffices to show
    \[
      w_{G'}(u',v')+\phi(u')-\phi(v')\ge 0.
    \]
    By definition, $\phi(z')=\dist_{\GC}(s',z')$ for $z'\in V(G')$. The triangle inequality in $\GC$ gives
    \[
      w_{\GC}(u',v') + \dist_{\GC}(s',u') - \dist_{\GC}(s',v') \;\ge\; 0.
    \]
    Since $w_{\GC}(u',v')=w_{G'}(u',v')$, we obtain $w_{G'}(u',v')+\phi(u')-\phi(v')\ge 0$, as desired.
\end{proof}

\paragraph{Correctness of the returned distances.}
\begin{lemma}
    $d_s(u)\ge \dist_H(s,u)$ for all $u\in V(H)$.
\end{lemma}
\begin{proof}
    By \Cref{lem:DAGsssp} and \Cref{lem:validDAGcall}, $d_{s''}(z)$ equals $\dist_{G''_\phi}(s'',z)$ for all $z\in V(G'')$. By \Cref{lem:potentialequivalence},
    \[
      d_{s''}(u'_i)+\phi(u'_i)=\dist_{G''}(s'',u'_i).
    \]
    Any $s''$-to-$u'_i$ path in $G''$ maps to an $s$-to-$u$ path in $H$ of the same weight by replacing each $u'_i$ with $\pi(u')$; the initial edge $(s'',u'_i)$ has weight $0$ and ensures the path starts at $s$. Taking minima over copies yields $d_s(u)\ge \dist_H(s,u)$.
\end{proof}

\begin{lemma}
    $d_s(u)\le \dist_H(s,u)$ for all $u\in V(H)$.
\end{lemma}
\begin{proof}
    Let $p$ be a shortest $s$-to-$u$ path in $H$. Since every edge weight is at least $-1$, we have the number of negative edges in $p$ is at most $k$ (by assumption) and hence $w_{\ge 0}(p)\le k$; otherwise $w(p)\ge -k+w_{\ge 0}(p)>0$, contradicting the $0$-edges from $s$.

    Recall $x=2\lambda$ and $\dcov=k/(2\lambda)$. Partition $p$ into at most $x$ subpaths $p=p_1\oplus\cdots\oplus p_x$ with $w_{\ge 0}(p_i)\le \dcov$. Each $p_i$ is covered by $G'_{\ge 0}$, hence also by $G'$ (weights restored). For each $i$, let $p'_i$ be a lift in $G'$, and let $p''_i$ be the copy of $p'_i$ in $G'_i$. By construction of $G''$, there is an edge from $\End(p''_i)$ to $\Start(p''_{i+1})$ with weight equal to $w(\End(p_i),\Start(p_{i+1}))$. Thus $p''=p''_1\oplus\cdots\oplus p''_x$ is an $s''$-to-$u'_x$ path in $G''$ with $w_{G''}(p'')=w(p)$. Therefore,
    \[
      \min_{\pi(u')=u}\dist_{G''}(s'',u'_i)\le \dist_H(s,u).
    \]
    Using $d_{s''}(u'_i)+\phi(u'_i)=\dist_{G''}(s'',u'_i)$ finishes the proof.
\end{proof}

\subsection{Time Complexity}

Let $T(m,k)$ denote the running time of $\SSSPFN(H,s,k)$. We show $T(m,n)=O(m\log^8 n)$ by writing a recurrence for $T(m,k)$.

\begin{itemize}
    \item Building the path cover for $H_{\ge 0}$ takes $O(m\log^4 n)$ time by \Cref{thm:pathcover} (since $\sqrt{\lambda}=100\log^3 n$).
    \item The recursive call requires an SCC decomposition in $O(m)$ time and costs
    \[
      T\!\left(\left(1+\tfrac{1}{\log n}\right)m,\ \tfrac{k}{2}\right).
    \]
    \item For computing distances, we build $G''$. Each vertex $u'_i$ has out-degree $O(\deg_H(\pi(u')))$, and there are $x=O(\log^6 n)$ copies, hence
    \[
      |E(G'')|\;\le\; O(\log^6 n)\cdot \sum_{u'\in V(G')}\deg_H(\pi(u')) \;=\; O(m\log^6 n)
    \]
    by \Cref{thm:pathcover}. Therefore \Cref{lem:DAGsssp} on $G''$ costs $O(m\log^7 n)$ time.
\end{itemize}
We obtain the recurrence
\[
T(m,k)\;\le\; T\!\left(\left(1+\frac{1}{\log n}\right)m,\frac{k}{2}\right)\;+\;O\!\left(m\log^7 n\right),
\]
with base case $T(m,O(\log^6 n))=O(m\log^7 n)$ by \Cref{lem:SSSPfewneg}. Solving gives $T(m,n)=O(m\log^8 n)$.

\bibliographystyle{alpha}
\bibliography{refs}

\appendix
\section{Basic Subroutines for SSSP}
\label{sec:basicSSSP}

\begin{proof}[Proof of \Cref{lem:DAGsssp}]
    Let the SCCs of $G$ be $(V_1,V_2,\dots,V_z)$ in a topological order of the condensation DAG. Define a potential function $\phi$ by
    \[
      \forall\, i\in[z],\ \forall\, v\in V_i:\quad \phi(v)=-W\cdot i.
    \]
    If $(u,v)\in E$ has $u,v$ in the same SCC, then $w_\phi(u,v)=w(u,v)\ge 0$. Otherwise, if $u\in V_i$ and $v\in V_j$ with $i<j$, then
    \[
      w_\phi(u,v)=w(u,v)+\phi(u)-\phi(v)=w(u,v)+W(j-i)\ge w(u,v)+W\ge 0.
    \]
    Hence all edges of $G_\phi$ have nonnegative weight, so Dijkstra's algorithm computes SSSP on $G_\phi$ in $O(m\log n)$ time. By \Cref{lem:potentialequivalence}, this yields a shortest-path tree for $G$ as well.
\end{proof}

\begin{proof}[Proof of \Cref{lem:SSSPfewneg}]
    Build a $(k+1)$-layer graph $G'$ as follows. For each $i\in[k+1]$, let $G_i$ be a copy of $G_{\ge 0}$ (the subgraph of nonnegative-weight edges), and let $u_i$ denote the copy of $u\in V$ in layer $i$. For every negative edge $(u,v)\in E$ with $w(u,v)<0$ and every $i\in[k]$, add the interlayer edge $(u_i,v_{i+1})$ with weight $w(u,v)$. Run Dijkstra's algorithm from $s_1$ in $G'$, and return
    \[
      \min_{i\in[k+1]}\ \dist_{G'}(s_1,u_i)\qquad\text{as }\dist_G(s,u).
    \]

    \emph{Running time.}
    We have $|V(G')|=(k+1)|V|$ and
    \[
      |E(G')|=(k+1)|E_{\ge 0}|+k|E_{<0}|\le O(km).
    \]
    Thus Dijkstra on $G'$ takes $O(km\log |V(G')|)=O(km\log n)$ time since $|V(G')|\le (k+1)n$.

    \emph{Correctness.}
    First, fix $i\in[k+1]$ and let $P'$ be any $s_1\!\to\!u_i$ path in $G'$. Projecting layers and replacing each interlayer edge by its original negative edge yields an $s\!\to\!u$ walk in $G$ of weight at most $w(P')$. Hence $\dist_{G'}(s_1,u_i)\ge \dist_G(s,u)$.

    Conversely, let $p$ be a shortest $s\!\to\!u$ path in $G$. By assumption, $p$ uses at most $k$ negative edges. Deleting these negative edges decomposes $p$ as
    \[
      p \;=\; p_1 \oplus p_2 \oplus \cdots \oplus p_z,
    \]
    with $z\le k+1$ and each $p_j$ consisting only of nonnegative edges. For each $j\in[z]$, let $p'_j$ be the corresponding path in layer $G_j$ (same vertices, same weight). For each negative edge between $p_j$ and $p_{j+1}$, use the interlayer edge from layer $j$ to $j+1$. Concatenating $p'_1,\ldots,p'_z$ with these interlayer edges yields an $s_1\!\to\!u_z$ path in $G'$ of weight exactly $w(p)=\dist_G(s,u)$. Therefore
    \[
      \min_{i}\dist_{G'}(s_1,u_i)\le \dist_G(s,u).
    \]
    Combining both directions gives $\dist_G(s,u)=\min_{i}\dist_{G'}(s_1,u_i)$.
\end{proof}

\section{Randomized or Undirected Solutions for \Cref{prob:path covering}}

\label{sec:previous solution}

\paragraph{Low Diameter Decomposition: Randomized Solution.}

Let $G$ be a directed graph $G$ with non-negative weight and $d'\ge0$. A \emph{low-diameter decomposition} (LDD) is an edge set $E'\subseteq E$ such that $G\setminus E'$ is a $\tilde{O}(d')$-clustered graph and, for each $e\in E$, $\Pr[e\in E']\le\frac{w(e)}{d'}\log n$. Bernstein et al\@.~\cite{BernsteinNW25} showed how to sample $E'$ in near-linear time. 

Observe that, by sampling $z=O(\log n)$ copies of LDD $E'_{i}$ of diameter parameter $d'=2d\log n$ and set $G'_{i}=G\setminus E'_{i}$, each length-$d$ path $p$ is covered by some $G'_{i}$ with probability at least $1-1/n^{10}$. This is because, for each $i$, 
\[
\Pr[p\cap E'_{i}\neq\emptyset]\le\sum_{e\in p}\Pr[e\in E'_{i}]\le\sum_{e\in p}\frac{w(e)}{d'}\log n\le\frac{w(p)}{d'}\log n\le1/2.
\]
So, there is $i$ where $p\subseteq G\setminus E'_{i}$ with probability $1-1/2^{z}\ge1-1/n^{10}.$ Thus, every path in the unknown target set is covered with probability at least $1-1/n^8$ by union bound.

\paragraph{Neighborhood Cover: Undirected Solution.}

A $d$-clustering ${\cal C}$ consists of vertex-disjoint subsets called \emph{clusters}, each of which has diameter at most $d$. For any vertex $v$, let $B_{v}(v,d)=\{w\mid\dist_{G}(v,w)\le d\}$. 

A $d$-neighborhood cover ${\cal N}_{d}$ of $G$ is a collection of $z=O(\log n)$ many $O(d\log n)$-clustering ${\cal C}_{1},\dots,{\cal C}_{z}$ such that, for every vertex $v$, $B_{v}(v,d')\subseteq S$ for some cluster $S$ in some clustering ${\cal C}_{i}$. For any $d$, Awerbuch~et~al.~\cite{awerbuch1998near} showed how to construct ${\cal N}_{d}$ deterministically in near-linear time. 

Observe that ${\cal N}_{d}$ gives the solution to \Cref{prob:path covering}. Indeed, for each $i\le z$, let $G'_{i}=\bigcup_{S\in{\cal C}_{i}}G[S]$. Consider any length-$d$ path $p$ in $G$ starting from vertex $u$. We have that $p\subseteq B_{v}(v,d)\subseteq G'_{i}$ for some $i$.

\section{A Barrier for Covering Paths via Subgraphs}
\label{sec:barrier}

In this section, we prove \Cref{thm:barrier}.
\barriar*

Fix $m,\lambda \ge 1$. We will build a directed graph $G$ with $O(m)$ edges and a distance parameter $d$.

\medskip
\textbf{Step 1: choose parameters.}
Let
\[
L := \bigl\lfloor \sqrt{m/\lambda} \bigr\rfloor
\quad\text{and}\quad
d := 2 + 3L .
\]
So $L = \Theta(\sqrt{m/\lambda})$ and $d = \Theta(\sqrt{m/\lambda})$.

We will also use one more parameter
\[
R := 2 d \lambda .
\]
Note that $R > d\lambda$ by construction.

\medskip
\textbf{Step 2: ladder gadget.}
Create vertices $a_1, a_2, \dots, a_{L+1}$. For each layer $j=1,\dots,L$ do:
\begin{itemize}
  \item create $R$ new vertices $w_{j,1},\dots,w_{j,R}$,
  \item add the directed cycle
  \[
    w_{j,1} \to w_{j,2} \to \cdots \to w_{j,R} \to w_{j,1},
  \]
  \item add edges $a_j \to w_{j,r}$ for every $r \in [R]$,
  \item add edges $w_{j,r} \to a_{j+1}$ for every $r \in [R]$.
\end{itemize}
Thus, to “go through” layer $j$ you can do
\[
a_j \to w_{j,r} \to w_{j,r+1} \to a_{j+1},
\]
which is 3 edges.

The number of edges in one layer is
\[
R \text{ (cycle)} + R \text{ (entries)} + R \text{ (exits)} = 3R = 3 \cdot 2 d \lambda = 6 d \lambda.
\]
Since there are $L$ layers, the total number of edges in the ladder is
\[
6 d \lambda \cdot L = 6 (2 + 3L) \lambda L = O(\lambda L^2) = O(\lambda \cdot (m/\lambda)) = O(m).
\]

\medskip
\textbf{Key property of the ladder.}
Let $H$ be any subgraph of $G$ whose SCCs all have diameter at most $d\lambda$. Consider a fixed layer $j$.
If $H$ contained all $R$ edges of the directed cycle of layer $j$, then that SCC would have directed diameter $R-1 \ge d\lambda$, contradicting that $H$ is $d\lambda$-clustered.
Hence:  For every $d\lambda$-clustered subgraph $H$ and every layer $j\in[L]$, $H$ omits at least one edge of that layer's cycle.

\medskip
\textbf{Step 3: star in front of the ladder.}
Add a special vertex $r$ (the center of the star) and add the edge $r \to a_1$.
Then add
\[
M := \bigl\lfloor m/20 \bigr\rfloor
\]
new vertices $s_1,\dots,s_M$ and add star edges
\[
s_t \to r \quad \text{for } t=1,\dots,M.
\]
This adds $M+1 = O(m)$ edges. Since the ladder already cost $O(m)$ edges, the whole graph $G$ has $O(m)$ edges.

\medskip
\textbf{Step 4: the $d$-paths we care about.}
Consider a path that:
\begin{enumerate}
  \item starts with some star edge $s_t \to r$,
  \item then takes $r \to a_1$,
  \item then for each layer $j=1,\dots,L$ takes a 3-edge traversal
  \[
  a_j \to w_{j,r_j} \to w_{j,r_j+1} \to a_{j+1}.
  \]
\end{enumerate}
This path has length
\[
1 \;(\text{star}) \;+\; 1 \;(\text{to } a_1) \;+\; 3L \;=\; 2 + 3L \;=\; d.
\]
So these are exactly $d$-paths in $G$.

\medskip
\textbf{Step 5: diagonalization for one star edge.}
Let $\{G'_1,\dots,G'_z\}$ be any collection of $d\lambda$-clustered subgraphs of $G$ that covers all $d$-paths.

Fix one star edge $e_t := s_t \to r$.
Define
\[
I_t := \{\, i \in [z] : e_t \in E(G'_i) \,\}.
\]
We claim that
\begin{equation}
\label{eq:It-large}
|I_t| \;\ge\; L.
\end{equation}

Suppose toward a contradiction that $|I_t| = q < L$.
For each $i \in I_t$, for every layer $j$ it omits at least one cycle edge in that layer.
Since $q<L$, we can pick for every $i \in I_t$ a \emph{distinct} layer $\ell(i) \in [L]$ and let $f_i$ be one cycle edge in layer $\ell(i)$ that $G'_i$ omits.

Now we build a $d$-path $P$ as in Step~4, but with the following extra choice:
in layer $\ell(i)$ we force the traversal to \emph{use} exactly the edge $f_i$, and in all other layers we traverse arbitrarily.
This is a valid $d$-path because every layer allows us to pick any two consecutive cycle vertices.

By construction:
\begin{itemize}
  \item If $i \notin I_t$, then $G'_i$ does not contain the first edge $e_t$, so $P \nsubseteq G'_i$.
  \item If $i \in I_t$, then $G'_i$ does not contain $f_i$ (by choice), and $P$ uses $f_i$, so again $P \nsubseteq G'_i$.
\end{itemize}
Thus $P$ is a $d$-path that is not contained in any $G'_i$, contradicting the assumption that the family covers all $d$-paths.
Therefore \eqref{eq:It-large} holds.

\medskip
\textbf{Step 6: counting.}
For each of the $M = \Theta(m)$ star edges $e_t$ we have found at least $L$ subgraphs that contain it.
So the total number of \emph{edge--subgraph incidences} is at least
\[
\sum_{t=1}^M |I_t| \;\ge\; M \cdot L.
\]
But
\[
\sum_{i=1}^z |E(G'_i)|
\]
counts all edge--subgraph incidences (each time an edge appears in some $G'_i$ it is counted once).
Hence
\[
\sum_{i=1}^z |E(G'_i)| \;\ge\; M \cdot L.
\]

Finally, $M = \Theta(m)$ and $L = \Theta\bigl(\sqrt{m/\lambda}\bigr)$, so
\[
\sum_{i=1}^z |E(G'_i)|
\;\ge\; \Omega(m) \cdot \Omega\bigl(\sqrt{m/\lambda}\bigr)
\;=\; \Omega\bigl(m \sqrt{m/\lambda}\bigr),
\]
which is exactly the desired lower bound.

This completes the proof.

\end{document}